\documentclass[12pt, draftclsnofoot, onecolumn]{IEEEtran}
\usepackage{graphicx}
\usepackage{amsthm}
\usepackage{epsfig}
\usepackage{latexsym}
\usepackage{amsfonts}
\usepackage{here}
\usepackage{rawfonts}
\usepackage[latin1]{inputenc}
\usepackage[T1]{fontenc}
\usepackage{calc}
\usepackage{capitalgreekitalic}
\usepackage{url}
\usepackage{enumerate}
\usepackage{color}
\usepackage[tbtags]{amsmath}
\usepackage{amssymb}
\usepackage{upref}
\usepackage{epic,eepic}
\usepackage{times}
\usepackage{dsfont}
\usepackage{comment}
\usepackage{cite}
\usepackage{bbm} 

\newcommand{\diag}{\ensuremath{\operatorname{diag}}}

\newtheorem{corollary}{Corollary}
\newtheorem{theorem}{\bf Theorem}
\newtheorem{proposition}{\bf Proposition}

\usepackage{dsfont}

\newcounter{step}
\newlength{\totlinewidth}
\newenvironment{algorithm}{%
  \rule{\linewidth}{1pt}
  \begin{list}{}%
    {\usecounter{step}%
      \settowidth{\labelwidth}{\textbf{Step 2:}}%
      \setlength{\leftmargin}{\labelwidth}%
      \setlength{\topsep}{-2pt}%
      \addtolength{\leftmargin}{\labelsep}%
      \addtolength{\leftmargin}{2mm}%
      \setlength{\rightmargin}{2mm}%
      \setlength{\totlinewidth}{\linewidth}%
      \addtolength{\totlinewidth}{\leftmargin}%
      \addtolength{\totlinewidth}{\rightmargin}%
      \setlength{\parsep}{0mm}%
      \raggedright}}%
  {\end{list}%
  \rule{\linewidth}{1pt}}
\newcounter{substep}

  {\end{list}}

\newlength{\aligntop}
\newlength{\alignbot}
 \makeatletter
\renewenvironment{align}{%
  \vspace{\aligntop}
  \start@align\@ne\st@rredfalse\m@ne
}{%
  \math@cr \black@\totwidth@
  \egroup
  \ifingather@
    \restorealignstate@
    \egroup
    \nonumber
    \ifnum0=`{\fi\iffalse}\fi
  \else
    $$%
  \fi
  \ignorespacesafterend%
  \vspace{\alignbot}\par\noindent
} \makeatother

\IEEEoverridecommandlockouts

\usepackage{algorithm}
\usepackage{algorithmic}

\usepackage{subfigure}

\begin{document}
\title{\LARGE Federated Echo State Learning for Minimizing Breaks in Presence in Wireless Virtual Reality Networks}


\author{{Mingzhe Chen}, Omid Semiari, \emph{Member, IEEE}, Walid Saad, \emph{Fellow, IEEE}, Xuanlin Liu, Changchuan Yin, \emph{Senior Member, IEEE}, \vspace*{-2em}\\ 

%
\thanks{M. Chen, X. Liu, and C. Yin are with the Beijing Key Laboratory of Network System Architecture and Convergence,
Beijing University of Posts and Telecommunications, Beijing, 100876, China, Emails: \protect{xuanlin.liu@bupt.edu.cn, ccyin@ieee.org}.}
\thanks{M. Chen is also with the Chinese University of Hong Kong, Shenzhen, 518172, China, and with the Department of Electrical Engineering, Princeton University, Princeton, NJ, 08544, USA, Email: \protect{mingzhec@princeton.edu}.}
\thanks{O. Semiari is with the Department of Electrical and Computer Engineering, University of Colorado Colorado Springs, Colorado Springs, CO, 80918, USA, Email: \protect{osemiari@uccs.edu}.}
\thanks{W. Saad is with the Wireless@VT, Bradley Department of Electrical and Computer Engineering, Virginia Tech, Blacksburg, VA, 24060, USA, Email: \protect{walids@vt.edu}.}
 }

\maketitle

%

\maketitle

\vspace{0cm}
\begin{abstract}
In this paper, the problem of enhancing the virtual reality (VR) experience for wireless users is investigated by minimizing the occurrence of breaks in presence (BIP) that can detach the users from their virtual world. To measure the BIP for wireless VR users, a novel model that jointly considers the VR application type, transmission delay, VR video quality, and users' awareness of the virtual environment is proposed. In the developed model, the base stations (BSs) transmit VR videos to the wireless VR users using directional transmission links so as to provide high data rates for the VR users, thus, reducing the number of BIP for each user. Since the body movements of a VR user may result in a blockage of its wireless link, the location and orientation of VR users must also be considered when minimizing BIP. 
The BIP minimization problem is formulated as  an optimization problem which jointly considers the predictions of users' locations, orientations, and their BS association.   
 To predict the orientation and locations of VR users, a distributed learning algorithm based on the machine learning framework of deep (ESNs) is proposed. The proposed algorithm uses concept from \emph{federated learning} to enable multiple BSs to locally train their deep ESNs using their collected data and cooperatively build a learning model to predict the entire users' locations and orientations. Using these predictions, the user association policy that minimizes BIP is derived.
  Simulation results demonstrate that the developed algorithm reduces the users' BIP by up to 16\% and 26\%, respectively, compared  to centralized ESN  and deep learning algorithms. 
\end{abstract}

\vspace{0cm}

{\renewcommand{\thefootnote}{\fnsymbol{footnote}}
\footnotetext{A preliminary version of this work was published in the IEEE GLOBECOM conference \cite{BIPchen}.}}
\section{Introduction}

{Deploying virtual reality (VR) applications over wireless networks is an essential stepping stone towards flexible deployment of pervasive VR applications \cite{saad2019vision}. However, to enable a seamless and immersive wireless VR experience, it is necessary to introduce novel wireless networking solutions that can meet stringent quality-of-service (QoS) requirements of VR applications \cite{bacstuug2016towards}.   
In wireless VR, any sudden drops in the data rate or increase in the delay can negatively impact the users' VR experience (e.g., due to interruptions in VR video streams). Due to such an interruption in the virtual world, VR users will experience \emph{breaks in presence (BIP)} events that can be detrimental to their immersive VR experience. While the fifth-generation (5G) new radio supports operation at high frequency bands as well as flexible frame structure to minimize latency, the performance of communication links at high frequencies is highly prone to blockage. That is, if an object blocks the wireless link between the BS and a VR user, the data rate can drop significantly and lead to a BIP. 
In addition to wireless factors such as delay and data rate, behavioral metrics related to each VR user such as the user's \emph{awareness} can also induce BIP. {\color{black}Awareness is defined as each wireless VR user's perceptions and actions in its individual VR environment.}  Therefore, to minimize the BIP of VR users, it is necessary to jointly consider all of the wireless environment and user-specific metrics that cause BIP, such as link blockage, user location, user orientation, user association, and user awareness. }

Recently, several works have studied a number of problems related to wireless VR networks\cite{7997740,sun2018communication,park2018urllc,8319985,kasgari2018human,8377419,8094203,8395443}. {\color{black}The work in\cite{7997740} developed a multipath cooperative route scheme to enable  VR wireless transmissions.} In \cite{sun2018communication}, the authors develop a framework for mobile VR delivery by leveraging the caching and computing capabilities of mobile VR devices. The authors in \cite{park2018urllc} study the problem of supporting visual and haptic perceptions over wireless cellular networks. A communications-constrained mobile edge computing framework is proposed in \cite{8319985} to reduce wireless resource consumption.
The work in \cite{kasgari2018human} proposes a concrete measure for the delay perception of VR users. The authors in \cite{8377419} present a scheme of proactive computing and high-frequency, millimeter wave (mmWave) \cite{8642794} transmission for wireless VR networks.
In \cite{8094203}, the authors design several experiments for quantifying the performance of tile-based $360^\circ$ video streaming over a real cellular network. Our previous work in \cite{8395443} studied the problems of resource allocation and $360^\circ$ content transmission. However, most of these existing works do not provide a comprehensive BIP model that accounts for the transmission delay, the quality of VR videos, VR application type, and user awareness. Moreover, the prior art in\cite{7997740,sun2018communication,park2018urllc,kasgari2018human,8094203,8395443,8319985,8377419} does not jointly consider the impact of the users' body movements when using mmWave communications. 


To address this challenge, machine learning techniques can be used to predict the users' movements and proactively determine the user associations that can minimize BIP.
 However, in prior works on machine learning for user movement predictions \cite{8372711,8219722,6503981,chen2016caching,8525324},
  the data for each user's movement must be collected by its associated BS.
 However, in real mobile VR scenarios, users will move and change their association and the data related to their movement is dispersed across multiple BSs. {\color{black}In such scenarios, the BSs may not be able to continuously share collected user data among each other}, due to the high overhead of data transmission. Moreover, sending all the information to a centralized processing server will cause very large delays that cannot be tolerated by VR applications. Thus, centralized machine learning algorithms such as in \cite{8372711,8219722,6503981,chen2016caching,8525324} will not be useful to predict real-time movements of the VR users.
To this end, \emph{a distributed learning framework that can be trained by the collected data at each BS is needed.} 

{\color{black} Recently, a number of existing works such as in \cite{bonawitz2019towards,konevcny2016federated,konevcny2015federated,samarakoon2018distributed} studied important problems related to the implementation of distributed learning over wireless networks. 
While interesting, these prior works \cite{bonawitz2019towards,konevcny2016federated,konevcny2015federated,samarakoon2018distributed} that focus on the optimization of the performance of distributed learning algorithms such as federated learning do not consider the use of distributed learning to optimize the performance of wireless networks. In particular, these existing works \cite{bonawitz2019towards,konevcny2016federated,konevcny2015federated,samarakoon2018distributed} do not consider the use of distributed learning algorithms to predict users' orientations and locations to reduce the BIP of wireless VR users. 
{\color{black}Note that, in \cite{BIPchen}, we have studied the use of a single-layer echo state network (ESN) model with federated learning for orientation and location predictions. However, the federated learning algorithm of \cite{BIPchen} cannot be used to analyze a large dataset. Meanwhile, the work in \cite{BIPchen} does not analyze the prediction accuracy or memory capacity of the introduced learning algorithm.}}

The key contribution of this work is to develop a novel framework for minimizing BIP within VR applications that operate over wireless networks. To our best knowledge, \emph{this paper is the {\color{black}first to analyze} how a wireless network with distributed learning can minimize BIP for VR users and enhance their virtual world experience}. The key contributions therefore include:

\begin{itemize}
\item For wireless VR users, we mathematically model a new BIP metric that jointly considers VR application type, the delay of VR video and tracking information transmission, VR video quality, and the users' awareness. 
\item To minimize the BIP of wireless VR users, we develop a federated ESN \cite{chen2017machine} learning algorithm that enables BSs to locally train their machine learning algorithms using the data collected from the users' locations and orientations. Then, the BSs can cooperatively build a learning model by sharing their trained models to predict the users' locations and orientations. Based on these predictions, we perform fundamental analysis to find an efficient  user association for each VR user that minimizes the BIP.

\item To analyze the prediction accuracy of the federated ESN learning algorithm, we study the memory capacity of federated ESNs. {\color{black}The memory capacity characterizes the ability of the ESN model to record historical locations and orientations of each VR user.} As the memory capacity increases, the prediction accuracy will improve. Since the BSs determine the user association based on these predictions, better prediction accuracy will lead to a more effective 
user association scheme that will minimize the number of BIP. The analytical results show that the memory capacity of ESNs depends on the number of neurons in each ESN model and the values of matrices that are used to generate the ESN model. 

\item Simulation results demonstrate that our proposed algorithm can achieve significant improvements in the statistics of BIP that occur within a wireless VR network. 
\end{itemize}
 
 The rest of this paper is organized as follows. The problem formulation is presented in Section \uppercase\expandafter{\romannumeral2}. The federated ESN learning algorithm for the predictions is proposed in Section \uppercase\expandafter{\romannumeral3}. In Section \uppercase\expandafter{\romannumeral4}, the memory capacity of various ESN models are analyzed. The user association is found in Section \uppercase\expandafter{\romannumeral5}. In Section \uppercase\expandafter{\romannumeral6}, simulation results are presented and conclusions are drawn in Section \uppercase\expandafter{\romannumeral7}.   

\section{System Model and Problem Formulation}\label{se:system}

\begin{figure}[!t]
  \begin{center}
   \vspace{0cm}
    \includegraphics[width=12cm]{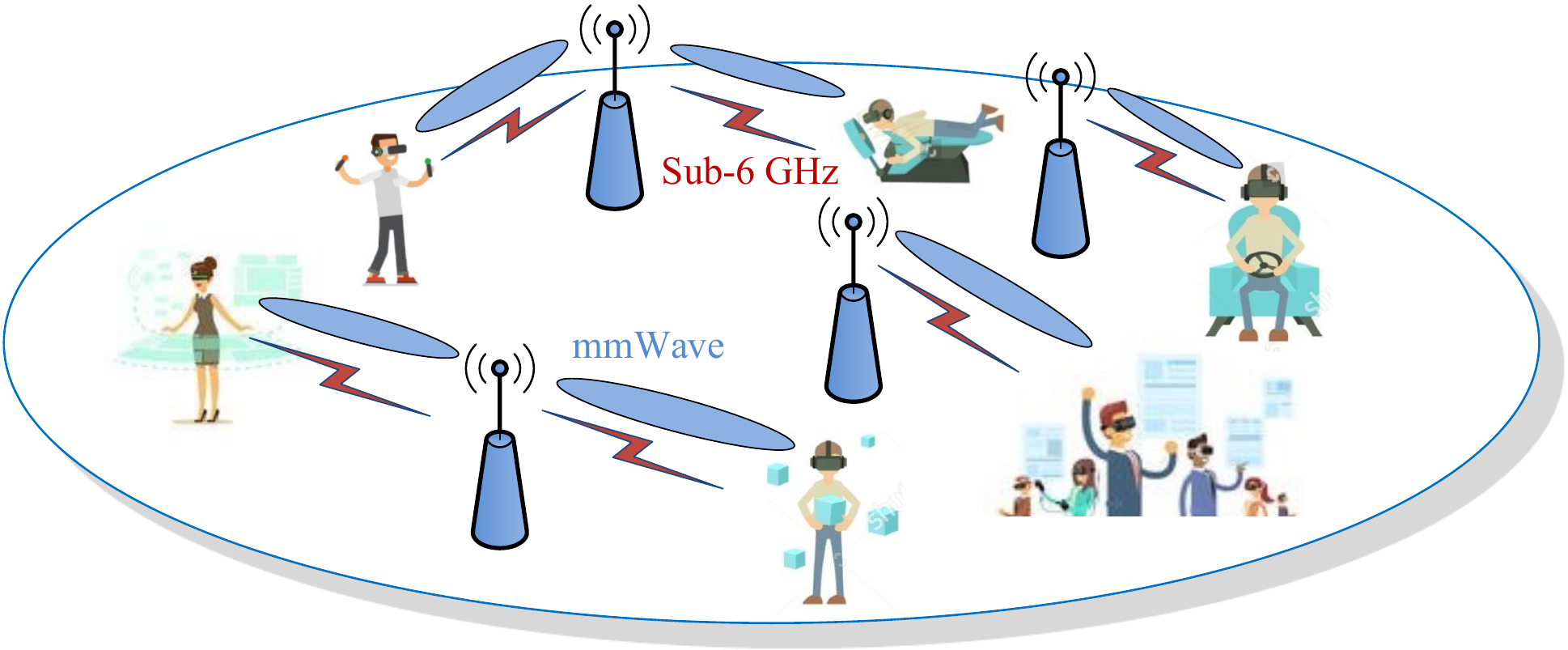}
    \vspace{-0.3cm}
    \caption{\label{architecture} The architecture of a wireless VR network. In this architecture, the Sub-6 GHz uplink is used to transmit tracking information and the mmWave downlink is used to transmit VR videos.}
  \end{center}\vspace{-1cm}
\end{figure}

Consider a cellular network that consists of a set $\mathcal{B}$ of $B$ BSs that service a set $\mathcal{U}$ of $U$ VR users. 
In this model, BSs act as VR controllers that can collect the tracking information related to the users' movements via VR sensors and use the collected data to generate the VR videos for their associated users, {\color{black}as shown in Fig. \ref{architecture}}. In particular, the uplink is used to transmit tracking information such as users' locations and orientations from the VR devices to the BSs, while the downlink is used to transmit VR videos from BSs to VR users. For user association, the VR users can associate with different BSs for uplink and downlink data transmissions.  
{\color{black}Different from prior works such as in \cite{bacstuug2016towards,sun2018communication,park2018urllc,8319985,kasgari2018human,8377419,8094203,8395443} that assume the VR users to be static, we consider a practical scenario in which the locations and orientations of the VR users will impact the VR application performance.} 

\subsection{Transmission Model} 
We consider both uplink and downlink transmission links between BSs and VR users. The VR users can operate at both mmWave and sub-6 GHz frequencies{\color{black}\cite{7929424,7398165,7381699}}. The VR videos are transmitted from BSs to VR users over the 28 GHz band. Meanwhile, the tracking information is transmitted from VR devices to their associated BSs over a sub-6 GHz frequency band. This is due to the fact that sub-6 GHz frequencies with limited bandwidth cannot support the large data rates required for VR video transmissions. However, it can provide reliable communications for sending small data sized users' tracking information. 

\subsubsection{Uplink Transmissions of User Tracking Information}
Let $\left(x_{it},y_{it}\right)$ be the Cartesian coordinates for the location of user $i$ at time $t$ and $S$ be the data size of each user's tracking information, including location and orientation. $S$ depends on the VR system (i.e., HTC Vive  \cite{htc} or Oculus \cite{oculus}). The data rate for transmitting the tracking information from VR user $i$ to BS $j$ is given by:

 \begin{equation} 
 c_{ij}^\textrm{UL}\left(x_{it},y_{it}\right)= {F^\textrm{UL}}{\log _2}\!\left(1+ {\frac{{{P_{u}}g_{ij}d_{ij}^{-\beta}\left(x_{it},y_{it}\right)}}{{\sum\limits_{k \in \mathcal{U}_i} {{P_{u}}g_{kj}d_{kj}^{-\beta}\left(x_{kt},y_{kt}\right)}  \!+\! {\rho  ^2}}}}\! \right)\!,
 \end{equation}
where $F^\textrm{UL}$ is the bandwidth of each subcarrier, $U_j^\textrm{UL}$ is the number of VR users associated with BS $j$ over uplink, $\mathcal{U}_{i}$ is the set of VR users that use the same subcarriers with user $i$, {$P_{u}$ is the transmit power of each VR user (assumed equal for all users)}, $g_{ij}$ is the Rayleigh channel gain, $d_{ij}$ is the distance between VR user $i$ and BS $j$ at time $t$, and ${\rho  ^2}$ is the noise power. 

\subsubsection{Downlink VR Video Transmission}
In the downlink, 
antenna arrays are deployed at BSs to perform directional beamforming over the mmWave frequency band. For simplicity, a sectored antenna model \cite{7934345} is used to approximate the actual array beam patterns. This simplified antenna model consists of four parameters: the half-power beamwidth $\phi$, the boresight direction $\theta$, the antenna gain of the mainlobe $Q$, and the antenna gain of the sidelobe $q$. Let $\varphi_{ij}$ be the phase from BS $j$ to VR user $i$. The antenna gain of the transmission link from BS $j$ to user $i$ is: 
\begin{equation}
{G_{ij}} = \left\{ {\begin{array}{*{20}{c}}
  {Q,~\textrm{if}~\left| {{\varphi _{ij}} - {\theta _j}} \right| \leqslant \frac{{{\phi }}}{2},} \\ 
  {q,~\textrm{if}~\left| {{\varphi _{ij}} - {\theta _j}} \right| > \frac{{{\phi }}}{2}.} 
\end{array}} \right.
\end{equation}

Since the VR device is located in front of the VR user's head, the mmWave link will be blocked, if the user rotates. Let ${\chi}_{it}$ be the orientation of user $i$ at time $t$ and $\vartheta$ be the maximum angle using which BS $j$ can directly transmit VR videos to a user without any human body blockage. $\phi'_{ij}$ denotes the phase from user $i$ to BS $j$. For user $i$, {\color{black}the blockage effect, ${b_{i}}\left( {\chi _{it}}\right)$, caused by its own body} can be given by:
\begin{equation}\label{eq:blockageeffect}
{b_{i}}\left( {\chi _{it}}\right) = \left\{ {\begin{array}{*{20}{c}}
  {1,~\textrm{if}~\left| {{\varphi'_{ij}} - {\chi _{it}}} \right| \leqslant {{{\vartheta }}},} \\ 
  {0,~\textrm{if}~\left| {{\varphi' _{ij}} - {\chi _{it}}} \right| >  {{{\vartheta}}}.} 
\end{array}} \right.
\end{equation}
We assume that each VR user's body constitutes a single blockage area and $n_{ijt}$ represents the number of VR users located between user $i$ and BS $j$ at time $t$. If there are no users located between user $i$ and BS $j$ that block the mmWave link, i.e., ($b_i\left({\chi _{it}} \right)+n_{ij}=0$, then, as shown in Fig. \ref{fig51}), the communication link between user $i$ and BS $j$ is line-of-sight (LoS). If the mmWave link between user $i$ and BS $j$ is blocked by the user $i$'s own body (as shown in Fig. \ref{fig52}, $b_i\left({\chi _{it}} \right)=1$) or blocked by other users located between user $i$ and BS $j$ (as shown in Fig. \ref{fig53}, $+n_{ij}>0$), then the communication link between user $i$ and BS $j$ is said to be non-line-of-sight (NLoS). {\color{black}From (\ref{eq:blockageeffect}), we can see that ${b_{i}}\left( {\chi _{it}}\right)$ and $n_{ij}$ can be directly determined by the users' orientations and locations.}
\begin{figure*}
\centerline{\subfigure[LoS links. ]{\includegraphics[width=4.5cm]{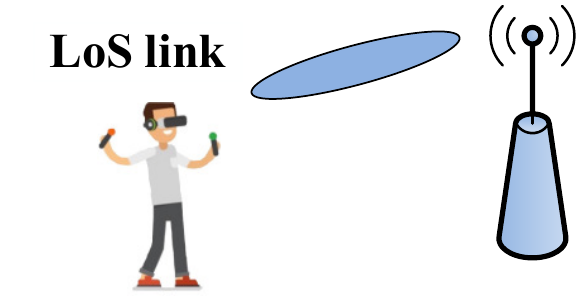}
\label{fig51}}\hspace{0.4 cm}
\subfigure[NLoS links caused by the user's own body.]{\includegraphics[width=4.7cm]{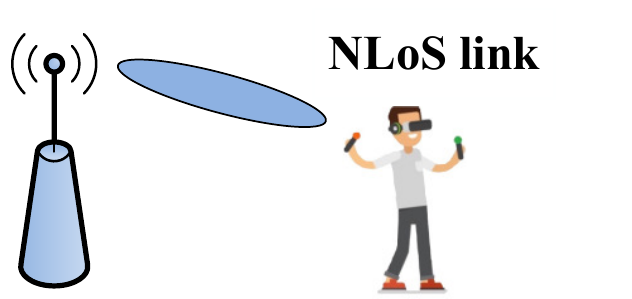}
\label{fig52}}\hspace{0cm}
\subfigure[NLoS links caused by other user's body. ]{\includegraphics[width=5.5cm]{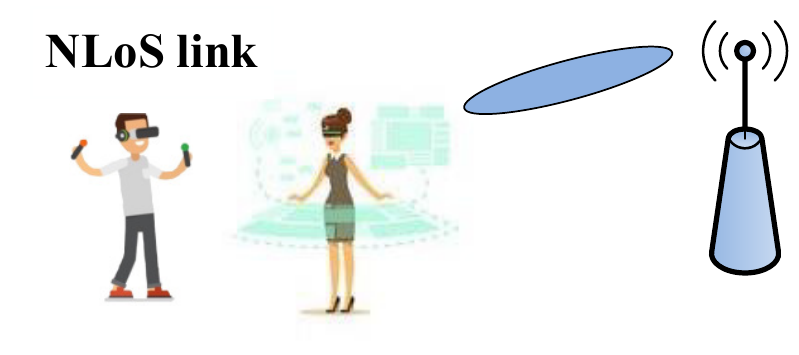}
\label{fig53}}}\vspace{-0.1cm}
\caption{\label{fig5}VR video transmission over LoS/NLoS links}\vspace{-1cm}
\end{figure*}

 Considering path loss and shadowing effects, the path loss for a LoS link and a NLoS link between VR user $i$ and BS $j$ in dB will be given by \cite{7934345}:
\begin{equation}
h_{ij}^\textrm{LoS}\left( {{x_{it}},{y_{it}}} \right) = 20\log \left(  \frac{d^0f_c4\pi}{\nu} \right)
+ 10{\color{black}\varpi_{\textrm{LoS}}}\log \left( {{d_{ij}}\left(  x_{it}, y_{it} \right)} \right) + {\mu _{\sigma_\textrm{LoS}} },
\end{equation}
\begin{equation}
h_{ij}^\textrm{NLoS}\left( {{x_{it}},{y_{it}}} \right) = 20\log \left(  \frac{d^0f_c4\pi}{\nu} \right)
+ 10{\color{black}\varpi_{\textrm{NLoS}}}\log \left( {{d_{ij}}\left(  x_{it}, y_{it} \right)} \right) + {\mu _{\sigma_\textrm{NLoS}} },
\end{equation}
where $20\log \left(  \frac{d_0f_c4\pi}{\nu} \right)$ is the free space path loss. Here, $d^0$ represents the reference distance, $f_c$ is the carrier frequency and $\nu $ is the light speed. {\color{black}$\varpi_\textrm{LoS}$} and {\color{black}$\varpi_\textrm{NLoS}$} represent the path loss exponents for the LoS and NLoS links, respectively. ${\mu_{\sigma_\textrm{LoS}} }$ and ${\mu_{\sigma_\textrm{NLoS}} }$ represent Gaussian random variables with zero mean, respectively. $\sigma_\textrm{LoS}$ and $\sigma_\textrm{NLoS}$ represent the standard deviations for LoS and NLoS links in dB, respectively. The downlink data rate of VR video transmission from BS $j$ to user $i$ is given by:
  \begin{equation}
   c_{ij}^\textrm{DL}\left(x_{it},y_{it}, b_i\left(\chi_{it}\right), n_{ij}\right)=\left\{ {\begin{array}{*{20}{c}}
  {{F^\textrm{DL}}{{\log }_2}\left( {1 + \frac{{{P_\textrm{B}G_{ij}}}}{{{{10}^{{{h_{ij}^\textrm{LoS}} \mathord{\left/
 {\vphantom {{h_{ij}^{LoS}} {10}}} \right.
 \kern-\nulldelimiterspace} {10}}}}{\rho  ^2}}}} \right),\textrm{if}~{b_i}\left( {{\chi _{it}}} \right) + {n_{ij}} = 0,} \\ 
  {{F^\textrm{DL}}{{\log }_2}\left( {1 + \frac{{{P_\textrm{B}G_{ij}}}}{{{{10}^{{{h_{ij}^\textrm{NLoS}} \mathord{\left/
 {\vphantom {{h_{ij}^{NLoS}} {10}}} \right.
 \kern-\nulldelimiterspace} {10}}}}{\rho  ^2}}}} \right),\textrm{if}~{b_i}\left( {{\chi _{it}}} \right) + {n_{ij}} > 0,} 
\end{array}} \right.
  \end{equation} 
  where $F^\textrm{DL}$ is the bandwidth allocated to each user and $P_\textrm{B}$ is the transmit power of each BS $j$ which is assumed to be equal for all BSs. {Since the downlink uses mmWave links, we assume that, due to directivity, interference in (6) can be neglected, as done in \cite{7491363}.} 
  
     
\subsection{Break in Presence Model}
In a VR application, the notion of a BIP represents an event that leads the VR users to realize that they are in a fictitious, virtual environment, thus ruining their immersive experience. In other words, a BIP event transitions a user from the immersive virtual world to the real world \cite{jerald2015vr}. For wired VR, BIP can be caused by various factors such as hitting the walls/ceiling, loss of tracking with the device, tripping on wire cords, or talking to another person from the real world \cite{jerald2015vr}. For wireless VR, BIP can be also caused by the delay of VR video and tracking information transmission, the quality of the VR videos received by the VR users, and inaccurate tracking information received by BSs. 

To model such BIP, we jointly consider the delay of VR video and tracking information transmission and the quality of the VR videos. We first define a vector $\boldsymbol{l}_{i,t}\left(c_{ij}^\textrm{DL}\left(x_{it},y_{it}, b_i\left(\chi_{it}\right), n_{ij}\right)\right)=\left[l_{i1,t},\ldots,l_{iN_L,t} \right]$ that represents  a VR video that user $i$ received at time $t$ with $l_{ik,t}\in \left\{0,1\right\}$. $l_{ik,t}=0$ indicates that pixel $k$ is not successfully received by user $i$, and $l_{ik,t}=1$, otherwise. We also define a vector $\boldsymbol{m}_{i,t}\left( G_A\right)=\left[m_{i1,t},\ldots, m_{iN_L,t} \right]^{\rm T}$ that represents the weight of the importance of each pixel constructing a VR video, where $m_{ik,t}\in \left[0,1\right]$ and $G_A$ represents a VR application such as an immersive VR game or a VR video. {\color{black} $m_{ik,t}=1$ indicates that pixel $k$ is one of the most important elements for the generation of $G_A$. The value of $m_{ik,t}$ depends on the compression used for the VR video. In each VR application $G_A$, a number of pixels can be compressed at the BS and recovered by the user and, hence, these pixels are not important. However, the pixels that cannot be compressed by the BS are important and must be transmitted to the VR users}. Therefore, each pixel will have different importance and $m_{ik,t}\in \left[0,1\right]$. Then, the BIP of VR user $i$ caused by the wireless transmission will be given by:  
{\color{black}
\begin{equation}\label{eq:BIP}
\begin{split}
\omega_{it}\left(x_{it},y_{it}, \chi_{it}, \boldsymbol{a}_{i,t}^\textrm{UL}, \boldsymbol{a}_{i,t}^\textrm{DL}\right)=&\\&\!\!\!\!\!\!\!\!\!\!\!\!\!\!\!\!\!\!\!\!\!\!\!\!\!\!\!\!\!\!\!\!\!\!\!\!\!\!\!\!\!\!\!\!\!\!\!\!\!\!\!\!\!\!\!\!\!\!\!\!\!\!\mathds{1}_{\left\{ \frac{A}{{a}_{ij,t}^\textrm{UL}c_{ij}^\textrm{UL}\left(x_{it},y_{it}\right)}+\frac{D\left( \boldsymbol{l}_{i,t}\left({a}_{ik,t}^\textrm{DL}c_{ik}^\textrm{DL}\left(x_{it},y_{it}, b_i\left(\chi_{it}\right), n_{ik}\right)\right)\right)}{{a}_{ik,t}^\textrm{DL}c_{ik}^\textrm{DL}\left(x_{it},y_{it}, b_i\left(\chi_{it}\right), n_{ik}\right)}  > \gamma_\textrm{D}~\vee~\boldsymbol{l}_{i,t}\left({a}_{ik,t}^\textrm{DL}c_{ik}^\textrm{DL}\left(x_{it},y_{it}, b_i\left(\chi_{it}\right), n_{ik}\right)\right)\boldsymbol{m}_{i,t}\left( G_A\right)  < \gamma_\textrm{Q} \right\}}.
\end{split}
\end{equation}}
where {\color{black}$\mathds{1}_{\left\{x\right\}}=1$ if $x$ is true, and otherwise, we have $\mathds{1}_{\left\{x\right\}}=0$.} $\mathds{1}_{\left\{x\right\}}  \vee  \mathds{1}_{\left\{y\right\}}=1$ as $y$ or $x$ is true, $\mathds{1}_{\left\{x\right\}} \vee  \mathds{1}_{\left\{y\right\}}=0$, otherwise. $\boldsymbol{a}_{i,t}^\textrm{UL}=\left[{a}_{i1,t}^\textrm{UL},\ldots, {a}_{iB,t}^\textrm{UL}\right]$ is a vector that represents user $i$'s uplink association with ${a}_{ik,t}^\textrm{UL} \in \left\{0,1\right\}$ and $\sum\limits_{k \in \mathcal{B}} {a}_{ik,t}^\textrm{UL}=1$. Similarly, $\boldsymbol{a}_{i,t}^\textrm{DL}=\left[{a}_{i1,t}^\textrm{DL},\ldots, {a}_{iB,t}^\textrm{DL}\right]$ is a vector that represents user $i$'s downlink association with ${a}_{ik,t}^\textrm{DL} \in \left\{0,1\right\}$ and $\sum\limits_{k \in \mathcal{B}} {a}_{ik,t}^\textrm{DL}=1$. $\gamma_D$ and $\gamma_Q$ represent the target delay and video quality requirements, respectively. In (\ref{eq:BIP}), {\color{black}$A$ represents the data size of the tracking information,} $\frac{A}{c_{ij}^\textrm{UL}\left(x_{it},y_{it}\right)}$ represents the time used for tracking information transmission from user $i$ to BS $j$, {\color{black}$D\left( \boldsymbol{l}_{i,t}\left(c_{ik}^\textrm{DL}\left(x_{it},y_{it}, b_i\left(\chi_{it}\right), n_{ik}\right)\right)\right)$ represents the data size of a VR video, and} $\frac{D\left( \boldsymbol{l}_{i,t}\left(c_{ik}^\textrm{DL}\left(x_{it},y_{it}, b_i\left(\chi_{it}\right), n_{ik}\right)\right)\right)}{c_{ik}^\textrm{DL}\left(x_{it},y_{it}, b_i\left(\chi_{it}\right), n_{ik}\right)} $ represents the transmission latency for sending the tracking information from BS $k$ to user $i$. For simplicity, hereinafter, {\color{black}$\omega_{it}$ is referred as $\omega_{it}\left(x_{it},y_{it}, \chi_{it}, \boldsymbol{a}_{i,t}^\textrm{UL}, \boldsymbol{a}_{i,t}^\textrm{DL}\right)$.}  
(\ref{eq:BIP}) shows that if the delay of VR video and tracking information transmission exceeds the target delay threshold allowed by VR systems or the quality of the VR video cannot meet the video requirement, users will experience a BIP ($\omega_{it}$=1).   
From (\ref{eq:BIP}), we can also see that, the BIP of user $i$ caused by wireless transmission depends on user $i$'s location, orientation, VR applications, and user association. {\color{black}(\ref{eq:BIP}) captures the BIP caused by wireless networking factors such as transmission delay and video quality. Next, we define a BIP model that jointly considers wireless transmission, the VR application type, and the users' awareness.} The BIP of user $i$ can be given by \cite{chung2010analysis}:  
\begin{equation}\label{eq:toBIP}
{\color{black}P_{i}\left(x_{it},y_{it}, G_A, \chi_{it}, \boldsymbol{a}_{i,t}^\textrm{UL}, \boldsymbol{a}_{i,t}^\textrm{DL}\right)= \frac{1}{T}\sum\limits_{t= 1}^{{T}} \left(G_A+\omega_{it}+G_A\omega_{it}+{\epsilon _{i}}+{\epsilon _{G_A|i}}+{\epsilon_B}\right),}
\end{equation}
where ${\epsilon _i}$ is user $i$'s awareness, $\epsilon _{G_A|i}$ is the joint effect caused by user $i$'s awareness and VR application $G_A$, and $\epsilon_B$ is a random effect. ${\epsilon _i}$, $\epsilon _{G_A|i}$, and $\epsilon_B$ follow a Gaussian distribution \cite{chung2010analysis} with zero mean and variances $\sigma _i^2$, $\sigma _{G_A|i}^2$, and $\sigma_B^2$, respectively. In (\ref{eq:toBIP}), the value of $P_{i}\left(x_{it},y_{it}, G_A, \chi_{it}, \boldsymbol{a}_{i,t}^\textrm{UL}, \boldsymbol{a}_{i,t}^\textrm{DL}\right)$ quantifies the average number of BIP that user $i$ can identify during a period.    From (\ref{eq:toBIP}), we can see that, as the VR application for user $i$ changes, the BIP value will change. For example, a given user watching VR videos will experience fewer BIP compared to a user engaged in an immersive first-person shooting game. This is due to the fact that in an immersive game environment, users are fully engaged with the virtual environment, as opposed to some VR applications that require the user to only watch VR videos. In (\ref{eq:toBIP}), we can also see that the BIP depend on the users' awareness. This means that different users will have different actions and perceptions when they interact with the virtual environment and, hence, different VR users may experience different levels of BIP. 

\subsection{Problem Formulation}
From (\ref{eq:toBIP}), we can see that the BIP of each user depends on this user's location, orientation, and selected BSs. By using an effective learning algorithm to predict the users' locations and orientations, the BSs can proactively determine the users' association to improve the downlink and uplink data rates and minimize BIP for each VR user.
The BIP minimization problem is:
 \addtocounter{equation}{0}
\begin{equation}\label{eq:max}
{\color{black}\mathop {\min }\limits_{\boldsymbol{a}_{i,t}^\textrm{UL}, \boldsymbol{a}_{i,t}^\textrm{DL}}\sum\limits_{i \in \mathcal{U} }P_{i}\left(\hat x_{it},\hat y_{it}, G_A, \hat \chi_{it}, \boldsymbol{a}_{i,t}^\textrm{UL}, \boldsymbol{a}_{i,t}^\textrm{DL}\right)}
\end{equation}
\begin{align}\label{c1}
\setlength{\abovedisplayskip}{-20 pt}
\setlength{\belowdisplayskip}{-20 pt}
&\rm{s.\;t.}\scalebox{1}{$\;\;\;\;\;{{U}_j } \leqslant  {V}, \;\;\;\;\; \forall j \in \mathcal{B},$}\tag{\theequation a}\\
&\scalebox{1}{$\;\;\;\;\;\;\;\;\;\; {a}_{ij,t}^\textrm{UL},{a}_{ij,t}^\textrm{DL}  \in \left\{0,1\right\},\;\;\;\forall i \in \mathcal{U},\forall j \in \mathcal{B}, $} \tag{\theequation b}\\
&\scalebox{1}{$\;\;\;\;\;\;\;\;\;\; \sum\limits_{j \in \mathcal{B}} {a}_{ij,t}^\textrm{UL}=1,\sum\limits_{j \in \mathcal{B}} {a}_{ij,t}^\textrm{DL}=1,\;\;\;\;\;\forall i \in \mathcal{U},$} \tag{\theequation c}
\end{align}
where $\hat x_{it}$, $\hat y_{it}$, and $\hat \chi_{it}$ are the predicted locations and orientation of user $i$ at time $t$, {\color{black}which depend on the actual historical locations and orientation of user $i$.}
${U}_j $ is the number of VR users associated with BS $j$ over downlink and $V$ is the maximum number of users that can be associated with each   BS. (\ref{eq:max}b) and (\ref{eq:max}c) show that each user can associate with only one uplink BS and one downlink BS. {From (\ref{eq:max}), we can see that the BIP of each user will depend on the user association as well as the users' locations and orientations. Meanwhile, the user association depends on the locations and orientations of the VR users. {\color{black}If the BSs perform the user association without knowledge of the locations and orientations of the users, the body blockage between the user-BS transmission links can potentially be significant, thus increasing the BIP of each user.} Therefore, the BSs must use historical information related to the users' locations and orientations to determine the user association. As the users' locations and orientations will continuously change as time elapses, BSs must proactively determine the user association to reduce the BIP of VR users. 
 In consequence, it is necessary to introduce a machine learning algorithm to predict the users' locations and orientations in order to determine the user association and minimize BIP of VR users.} {In the model defined in Section \ref{se:system}, the user association changes as the users' location and orientation vary with time. Consequently, each BS that connects to a given VR user can only collect partial information about this user's locations and orientation. However, a BS cannot rely on partial information to predict each user's location and orientation. Moreover, since a given VR user will change its association, the data pertaining to this VR user's movement will be located at multiple BSs.  Hence, traditional centralized learning algorithms that are implemented by a given BS
 cannot predict the entire VR user's locations and orientations without knowing the user's data collected by other BSs. 
To overcome the challenges mentioned previously, we introduce a distributed federated learning framework that can predict the location and orientation of each VR user as the training data related to each user's locations and orientations is located at multiple BSs.} 

\section{Federated Echo State Learning for Predictions of the Users' location and Orientation}\label{se:esn}

\emph{Federated learning} is a decentralized learning algorithm \cite{amiri2018computation} that can operate by using training datasets that are distributed across multiple devices (e.g., BSs). For our system, one key advantage of federated learning is that it can allow multiple BSs to locally train their local learning model using their collected data and cooperatively build a learning model by sharing their locally trained models.
 Compared to existing federated learning algorithms \cite{smith2017federated} that use matrices to record the users' behavior and cannot analyze the correlation of the users' behavior data, we propose an ESN-based federated learning algorithm that can use an ESN to efficiently analyze the data related to the users' location and orientation. {\color{black}The proposed algorithm enables the BSs to collaboratively generate a global ESN model to predict the whole set of locations and orientations for each user without transmitting the collected data to other BSs. 
 However, if the BSs use the centralized learning algorithms for the orientation and location predictions, the BSs must use the data collected from all BSs to train the algorithm. ESNs have two unique advantages: simple training process and the ability to analyze time-dependent data \cite{chen2017machine}. Since the data that is related to the orientation and locations of the users is time-dependent and the users' orientation and locations will change frequently, we must use ESNs that can efficiently analyze time-dependent data and converge quickly to obtain the prediction results on time and determine the user association.} 
Next, we first introduce the components of the federated ESN learning model. Then, we explain the entire procedure of using our federated ESN learning algorithm to predict the users' locations and orientation.

\subsection{Components of Federated ESN Learning Algorithm}\label{se:CESNFL}
A federated ESN learning algorithm consists of five components: {\color{black}a) agents, b) input, c) output, and d) local ESN model}, which are specified as follows:
 
 \begin{itemize}
\item \emph{Agent}: In our system, we need to define an individual federated ESN learning algorithm to predict the location and orientation of each VR user. Meanwhile, each user's individual federated ESN learning algorithm must be implemented by all BSs that have been associated with this user. {\color{black}Each BS $j$ must implement $U$ learning algorithms to predict the locations and orientations of all users.}

\item \emph{Input:} {\color{black}The input of the federated ESN learning algorithm that is implemented by BS $j$ for the predictions of each VR user $i$ is defined by a vector $\boldsymbol{\upsilon}_{ij}=\left[ {\boldsymbol{\upsilon}_{ij,1}, \cdots , \boldsymbol{\upsilon}_{ij,T}} \right]^{\mathrm{T}}$ that represents the information related to user $i$'s location and orientation where $\boldsymbol{\upsilon}_{ij,t}=\left[{\xi }_{ij1,t},\ldots, {\xi}_{ijN_x,t}\right]$ represents user $i$'s information related to location and orientation at time $t$. This information includes user $i$'s locations, orientations, and VR applications.} $N_x$ is the number of properties that constitute a vector $\boldsymbol{\upsilon}_{ij,t}$. The input of the proposed algorithm will be combined with the ESN model to predict users' orientation and locations. BSs will use these predictions to determine user associations.


\item \emph{Output:} For each user $i$, the output of the federated ESN learning algorithm at BS $j$ is a vector $\boldsymbol{y}_{ij,t}= \left[ { {\hat{\boldsymbol{y}}_{ijt+1}}, \ldots , { \hat{\boldsymbol{y}}_{ijt+Y}}  } \right]$ of user $i$'s locations and orientations where ${\hat{\boldsymbol{y}}_{ijt+k}}=\left[ \hat x_{it+k}, \hat y_{it+k}, \hat \chi_{it+k}\right]$ with $ \hat x_{it+k}$ and $\hat y_{it+k}$ being the predicted location coordinates of user $i$ at time $t+k$ and $\hat \chi_{it+k}$ being the estimated orientation of user $i$ at $t+k$. $Y$ is the number of future time slots that a federated ESN learning algorithm can predict. The predictions of the locations and orientations can be used to determine the user's association. 
\item \emph{Local ESN model:} For each BS $j$, a local ESN model is used to build the relationship between the input of all BSs and the predictions of the users' location and orientation, as shown in Fig. \ref{f:LSMmodel}. The local ESN model consists of the input weight matrix $\boldsymbol{W}_{j}^\textrm{in} \in {\mathbb{R}^{N_W \times T}}$, recurrent matrix $\boldsymbol{W}_{j} \in {\mathbb{R}^{N_W \times N_W}}$, and the output weight matrix $\boldsymbol{W}_{j}^\textrm{out} \in {\mathbb{R}^{Y \times \left(N_W+T\right)}}$. The values of $\boldsymbol{W}_{j}^\textrm{in}$ and $\boldsymbol{W}_{j}$ are generated randomly. However, the output weight matrix $\boldsymbol{W}_{j}^\textrm{out}$ need to be trained according to the inputs of all BSs.  

\begin{figure}
\centering
\vspace{0cm}
\subfigure[a series ESN model]{
\label{fig7a} 
\includegraphics[width=9cm]{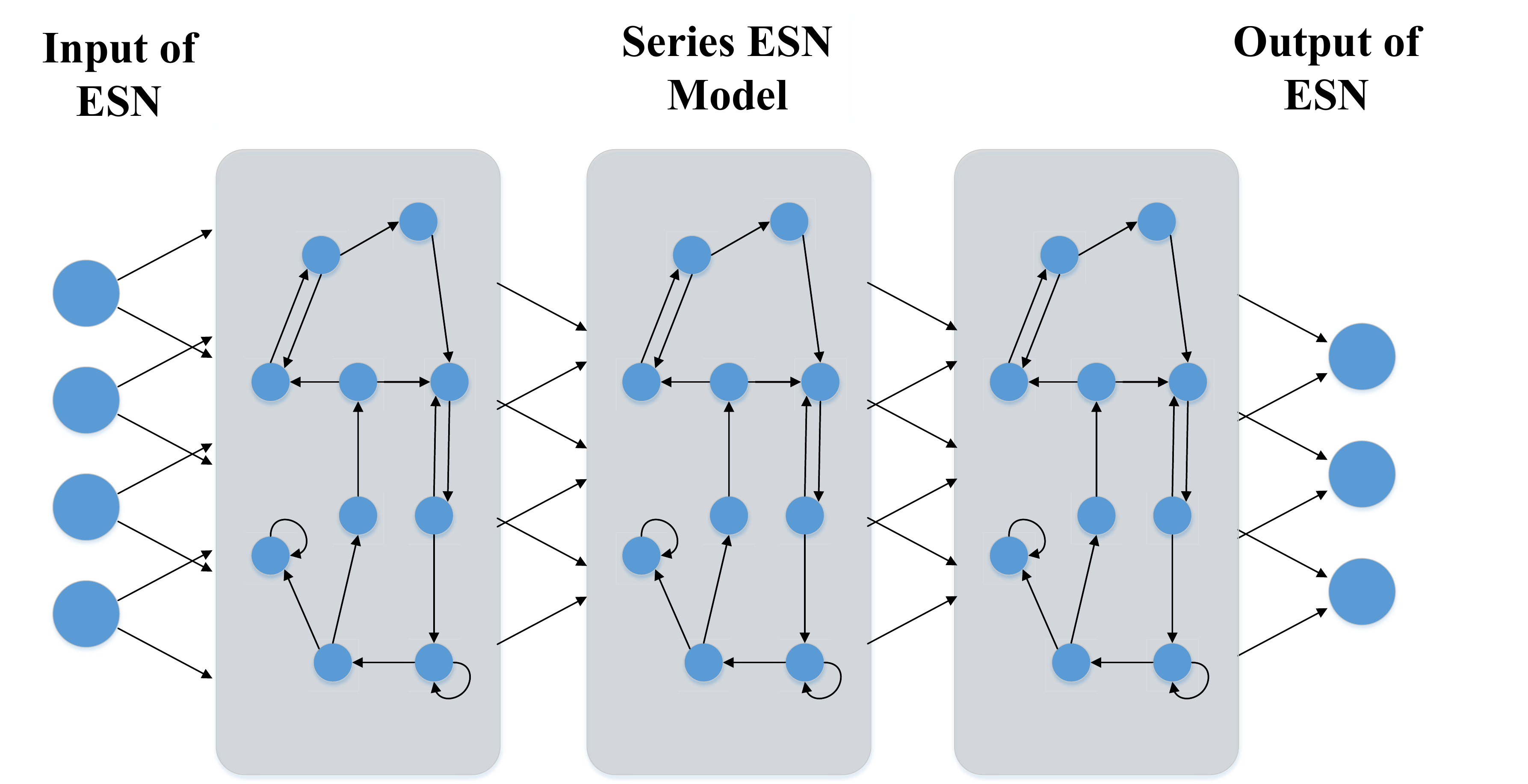}}
\subfigure[a parallel ESN model]{ 
\label{fig7b} 
\includegraphics[width=5cm]{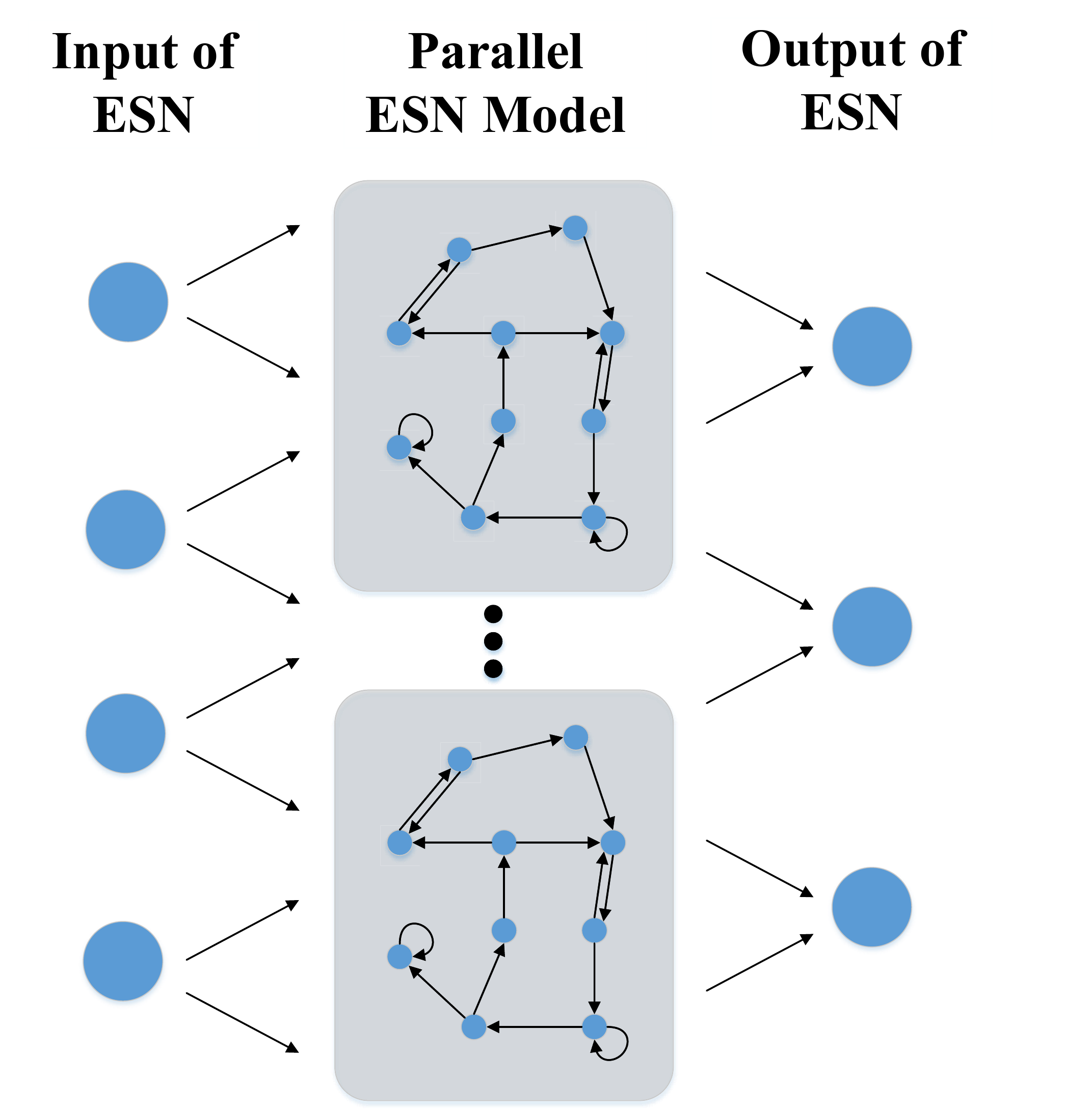}}
  \vspace{-0.2cm}
 \caption{\label{fig7} {Architectures of deep ESN models}.}
 \vspace{-0.6cm}
\end{figure}

We introduce three ESN models: \emph{single ESN model,  series ESN model, and  parallel ESN model}. In the single ESN model, an ESN is directly connected to the input and output. Moreover, as shown in Fig. \ref{fig7}, series and parallel ESN models connect single ESN models in series and parallel, respectively. {\color{black}Each ESN model has its own advantage for our problem. In particular, a single ESN model can converge faster than a series ESN model and a parallel ESN model. A parallel ESN model has a larger memory capacity than a series ESN model. A series ESN model can decrease the prediction errors in the training process.} 

\end{itemize}

 \begin{figure}[!t]
  \begin{center}
   \vspace{0cm}
    \includegraphics[width=13cm]{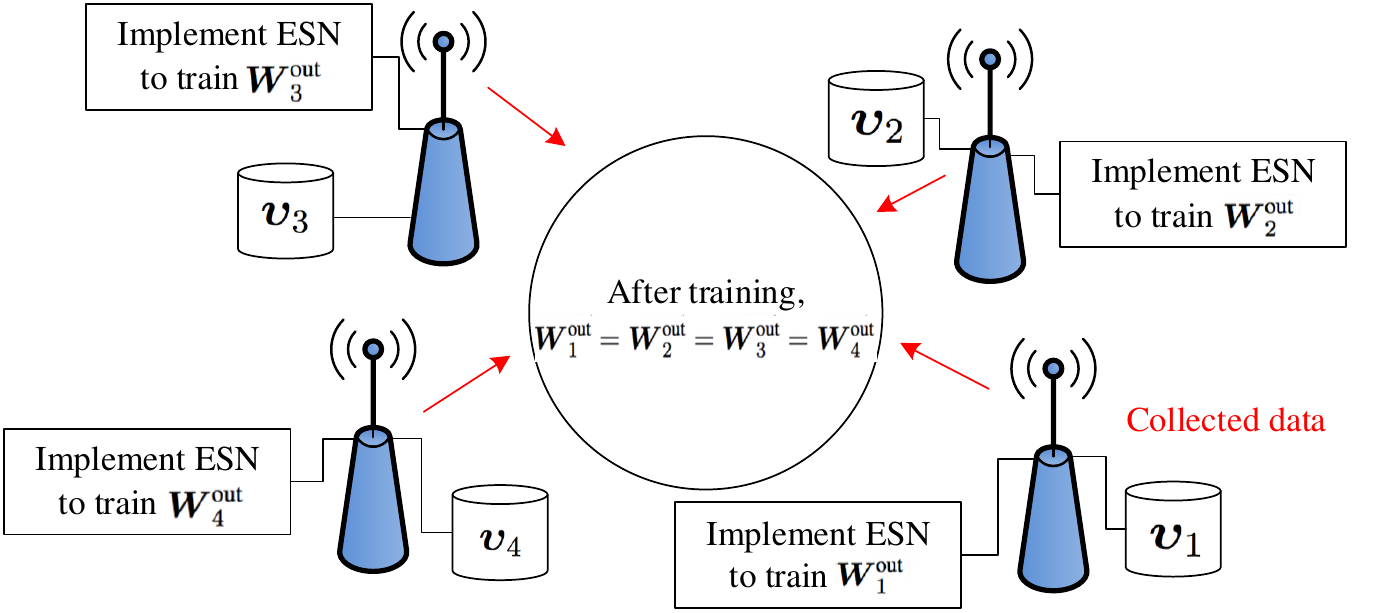}
    \vspace{-0.3cm}
    \caption{\label{f:LSMmodel} The implementation of the ESN based federated learning. {\color{black}Here, the data is located at the BSs and
the learning model $\boldsymbol{W}_j^{\textrm{out}}$ that is trained by each BS's collected data is the local model.} }
  \end{center}\vspace{-1.2cm}
\end{figure}

\subsection{ESN Based Federated Learning Algorithm for Users' Location and Orientation Predictions}
Next, we explain the entire procedure of training the proposed ESN-based federated learning algorithm. Our purpose of training ESN is to find an optimal output weight matrix in order to accurately predict the users' locations and orientations, as shown in Fig. \ref{f:LSMmodel}. 

To introduce the training process, we first explain the state of the neurons in ESN. The neuron states of the proposed algorithm implemented by BS $j$ for the predictions of user $i$ are:
 \begin{equation}\label{eq:state}
{\boldsymbol{\mu}_{j,t}} = {\boldsymbol{W}_j{\boldsymbol{\mu}_{j,t-1}} + \boldsymbol{W}_j^\textrm{in}{\boldsymbol{\upsilon}_{ij,t}}}.
\end{equation} 
  Based on the states of neurons and the inputs, the ESN can estimate the output, which is:
 \begin{equation}\label{eq:update}
\hat{\boldsymbol{y}}_{ij,t} = {\boldsymbol{W}_{j,t}^\textrm{out}} \left[ {\begin{array}{*{20}{c}}
    {\boldsymbol{\upsilon}_{ij,t}} \\ 
 {\boldsymbol{\mu}_{j,t}} 
\end{array}} \right]. 
\end{equation}
From (\ref{eq:update}), we can see that, in order to enable an ESN to predict the users' locations and orientations, we only need to adjust the value of the output weight matrix. However, each BS can collect only partial data for each user and, hence, we need to use a distributed learning algorithm to train the ESNs. To introduce the distributed learning algorithm, we first define two matrices which are given by: 
\begin{equation}
\boldsymbol{H}_j=\left[ {\begin{array}{*{20}{c}}
  {\begin{array}{*{20}{c}}
  {{\boldsymbol{\upsilon}_{ij,1}}}& {\boldsymbol{\mu}_{j,1}} 
\end{array}} \\ 
   \vdots  \\ 
  {\begin{array}{*{20}{c}}
  {{\boldsymbol{\upsilon}_{ij,T}}}&{\boldsymbol{\mu}_{j,T}} 
\end{array}} 
\end{array}} \right]~\textrm{and}~
\boldsymbol{E}_j=\left[ \boldsymbol{e}_{ij,1}, \ldots, \boldsymbol{e}_{ij,T}\right],
\end{equation}
 where $ \boldsymbol{e}_{ij,t}$ is the desired locations and orientations of each VR user, given the ESN input $  {{\boldsymbol{\upsilon}_{ij,t}}}$. Then, the training purpose can be given as follows:
\begin{equation}\label{eq:max2}
\mathop {\min }\limits_{\boldsymbol{W}^\textrm{out}}\frac{1}{2}\left( {\sum\limits_{j = 1}^B {\left\| {{\boldsymbol{W}^\textrm{out}}\boldsymbol{H}_j^{\rm T} - {\boldsymbol{E}_j}} \right\|^2} } \right)+\frac{\lambda}{2}{\left\| {{\boldsymbol{W}^\textrm{out}}} \right\|}.
\end{equation}
{\color{black}(\ref{eq:max2}) is used to find the optimal global output weight matrix ${\boldsymbol{W}^\textrm{out}}$} according to which the BSs can predict the entire users' locations and orientations without the knowledge of the users' data collected by other BSs. From (\ref{eq:max2}), we can see that, {\color{black}each BS $j$ needs to adjust its output weight matrix $\boldsymbol{W}_j^\textrm{out}$ and find the optimal output weight matrix $\boldsymbol{W}^\textrm{out}$. After the
learning step, we have $\boldsymbol{W}_j^\textrm{out}=\boldsymbol{W}^\textrm{out}$, which means that when the learning algorithm converges, the local model of each BS will converge to the global model.}. {\color{black}A standard update policy of $\boldsymbol{W}_j^\textrm{out}$ for the augmented Lagrangian problem in (\ref{eq:max2}) is given by \cite{scardapane2016decentralized}: }
\begin{equation}\label{eq:Wjout}
\boldsymbol{W}_{j,t+1}^\textrm{out}=\varsigma ^{-1} \left[\boldsymbol{I}-\boldsymbol{H}_j^{\rm T}\left( \varsigma  \boldsymbol{I}+ \boldsymbol{H}_j \boldsymbol{H}_j^{\rm T}  \right)\boldsymbol{H}_j^{\rm T} \right]\left(\boldsymbol{H}_j^{\rm T} \boldsymbol{E}_j-\boldsymbol{n}_{j,t}+\varsigma \boldsymbol{W}_{t}^\textrm{out} \right),
\end{equation}
where $\varsigma$ is the learning rate and $\boldsymbol{W}_{t}^\textrm{out} $ is the optimal output weight matrix that the ESN model of each BS needs to find. From (\ref{eq:Wjout}), we can see that $\boldsymbol{W}_{j,t+1}^\textrm{out}$ is the output weight matrix that is generated at BS $j$. $\boldsymbol{W}_{j,t+1}^\textrm{out}$ can only be used to predict partial locations and orientations given the users' data collected by BS $j$. $\boldsymbol{W}_{j,t+1}^\textrm{out}$ is different from the output weight matrices of other BSs.
The optimal output weight matrix is given by:
\begin{equation}\label{eq:Wt}
\boldsymbol{W}_{t+1}^\textrm{out}=\frac{B \varsigma \hat{\boldsymbol{W}}_{t+1}^\textrm{out}+B\hat{\boldsymbol{n}}_{t}}{\lambda+\varsigma B},
\end{equation}
where  $\hat{\boldsymbol{W}}_{t+1}^\textrm{out}$ and $\hat{\boldsymbol{n}}_{t+1}^\textrm{out}$ can be calculated as follows: 
\begin{equation}\label{eq:Wout}
\hat{\boldsymbol{W}}_{t+1}^\textrm{out}=\frac{1}{B}\sum\limits_{j = 1}^B {\boldsymbol{W}_{j,t + 1}^\textrm{out}},~\hat{\boldsymbol{n}}_{t}=\frac{1}{B}\sum\limits_{j = 1}^B {\boldsymbol{n}_{j,t }}.
\end{equation}
{\color{black}From (\ref{eq:Wjout}) to (\ref{eq:Wout}), we can see that the global output weight matrix $\boldsymbol{W}^\textrm{out}$ is based on (\ref{eq:Wt}) and (\ref{eq:Wout}) while the local output weight matrix $\boldsymbol{W}_{j}^\textrm{out}$ is based on (\ref{eq:Wjout}). } In (\ref{eq:Wjout}), $\boldsymbol{n}_{j,t}$ is the deviation between the output weight matrix $\boldsymbol{W}_{j,t+1}^\textrm{out}$ of each BS $j$ and the optimal output weight matrix $\boldsymbol{W}_{t+1}^\textrm{out} $ that the ESN model of each BS needs to converge, which is given by:
\begin{equation}\label{eq:nj}
\boldsymbol{n}_{j,t+1}=\boldsymbol{n}_{j,t}+\gamma\left(\boldsymbol{W}_{j,t+1}^\textrm{out}-\boldsymbol{W}_{t+1}^\textrm{out} \right).
\end{equation}
$\boldsymbol{W}_{t+1}^\textrm{out}$ is the global optimal output weight matrix that can be used to predict the entire locations and orientations of a given user. This means that using $\boldsymbol{W}_{t+1}^\textrm{out}$, each BS can predict the entire user's locations and orientations as the BS only collects partial data related to the user's locations and orientations. As time elapses, $\boldsymbol{W}_{j,t+1}^\textrm{out}$ will finally converge to $\boldsymbol{W}_{t+1}^\textrm{out}$. In consequence, all of BSs can predict the entire locations and orientations of each user. {To measure the convergence, we define two vectors which can be given by $\boldsymbol{r}_{j,t}=\boldsymbol{W}_{j,t}^\textrm{out}-\boldsymbol{W}_{t}^\textrm{out}$ and $\boldsymbol{s}_{j,t}=\boldsymbol{W}_{t}^\textrm{out}-\boldsymbol{W}_{t-1}^\textrm{out}$. As $\left\|\boldsymbol{r}_{j,t+1}\right\| \leqslant \gamma_A$ or $\left\|\boldsymbol{s}_{j,t}\right\|\leqslant \gamma_A$, the proposed algorithm converges. $\gamma_A$ is determined by the BSs. {\color{black} Since the minimization function in (\ref{eq:max2}) is a convex function, the BSs are guaranteed to find an optimal output weight matrix that satisfy $\left\|\boldsymbol{r}_{j,t+1}\right\| \leqslant \gamma_A$ or $\left\|\boldsymbol{s}_{j,t}\right\|\leqslant \gamma_A$.} As $\gamma_A$ increases, the accuracy of the predictions and the number of iterations decrease. Therefore, BSs need to jointly account for the time used for training ESN and the prediction accuracy to determine the value of $\gamma_A$. In fact, the ESN } As the learning algorithm converges, each BS can use its own ESN to predict the entire location and orientation of each VR user. According to these predictions, BSs can determine the user association to minimize the BIP of VR users. Algorithm \ref{algorithm} summarizes the entire process of using ESN based federated learning algorithm for the predictions of the users' locations and orientations. From Algorithm \ref{algorithm}, we can see that $\boldsymbol{W}_{j}^\textrm{in}$ and $\boldsymbol{W}_{j}^\textrm{out}$ are local parameters which means that each BS $j$ will generate its own $\boldsymbol{W}_{j}^\textrm{in}$ and $\boldsymbol{W}_{j}^\textrm{out}$. However, $\boldsymbol{W}_{j}$ is a global parameter which means that all of the BSs will have the same $\boldsymbol{W}_{j}$.   

\begin{algorithm}[t]\footnotesize
\caption{ Federated ESN learning algorithm for location and orientation predictions}   
\label{algorithm}   
\begin{algorithmic} [1] 
\REQUIRE Training data set (local), $\boldsymbol{\upsilon}_{ij}$.\\ 
\vspace{1pt}  
\ENSURE Each BS $j$ generates the ESN model for each user including $\boldsymbol{W}_{j}^\textrm{in}$ (local), $\boldsymbol{W}_{j}$ (global), and $\boldsymbol{W}_{j}^\textrm{out}$ (local). \\ 
\vspace{1pt}  
\STATE Obtain the matrices $\boldsymbol{H}_j$ and $\boldsymbol{E}_j$ based on (\ref{eq:state}).
\vspace{1pt}  
\FOR {time $t$}
\vspace{1pt}
\STATE  Compute $\boldsymbol{W}_{j,t+1}^\textrm{out}$ using (\ref{eq:Wjout}).    
\vspace{1pt}  
\STATE Calculate $\hat{\boldsymbol{W}}_{t+1}^\textrm{out}$ and $\hat{\boldsymbol{n}}_{t}^\textrm{out}$ based on (\ref{eq:Wout}).   
\vspace{1pt}  
\STATE Calculate $\boldsymbol{W}_{t+1}^\textrm{out}$ based on (\ref{eq:Wt}).  
\vspace{1pt} 
\STATE Compute $\boldsymbol{n}_{j,t+1}$ based on (\ref{eq:nj}). 
\vspace{1pt} 
\STATE Compute $\left\|\boldsymbol{r}_{j,t+1}\right\|$ and $\left\|\boldsymbol{s}_{j,t}\right\|$.
\vspace{1pt} 
\STATE If $\left\|\boldsymbol{r}_{j,t+1}\right\| \leqslant \gamma_A$ or $\left\|\boldsymbol{s}_{j,t}\right\|\leqslant \gamma_A$, the algorithm converges.
\ENDFOR  
\end{algorithmic}
\end{algorithm}  

\section{Memory Capacity Analysis } \label{se:mc}     
To improve the prediction accuracy of the proposed algorithm, we analyze the memory capacity of the proposed ESN model. The memory capacity quantifies the ability of each ESN to record the historical locations and orientations of each VR user. As the memory capacity of the ESNs increases, the ESNs can record more historical data related to users' locations and use this information to achieve better prediction\footnote{{\color{black}Here, as the size of the recorded data increases, the ESNs can use more historical data to build a relationship between historical orientations and locations, and future orientations and locations. Hence, the ESN prediction accuracy improves.}} for the users' locations and orientations. {\color{black} The analysis of the ESN memory capacity will be used for the choice of the ESN models for the predictions of the users' locations and orientations. } 
Next, we derive closed-form expressions of the memory capacity of the three ESN models that we described in Section \ref{se:esn}, namely, the single ESN model, the parallel ESN model, and the series ESN model. {\color{black}Note that, our previous work \cite{8614097} analyzed the memory capacity for a centralized parallel ESN model. In contrast, here, we analyze the memory capacity for three ESN models used for federated learning.}   

%

 We assume that the input of each ESN model at time $t$ is $m_t$ and the output of each ESN model is $z_t$. Then, the memory capacity of each ESN model is given by \cite{7880663}: 
\begin{equation}\label{eq:memory}
	M = \sum_{k=1}^{\infty}\frac{\textrm{Cov}^2(m_{t-k},z_t)}{\textrm{Var}(m_t)\textrm{Var}(z_t)},
\end{equation} 
where {\color{black} \textrm{Cov} and \textrm{Var} represent the covariance and variance operators, respectively.}  {\color{black} In (\ref{eq:memory}), $\frac{\textrm{Cov}^2(m_{t-k},z_t)}{\textrm{Var}(m_{t-k})\textrm{Var}(z_t)}$ captures the correlation between the ESN input $m_{t-k}$ at time $t-k$ and the ESN output $z_t$ at time $t$. $\frac{\textrm{Cov}^2(m_{t-k},z_t)}{\textrm{Var}(m_{t-k})\textrm{Var}(z_t)}=1$ indicates that $m_{t-k}$ and $z_t$ are related which means that the output $z_t$ includes the information of $m_{t-k}$ and, hence, the ESN can record input $m_{t-k}$ at time $t$. $\frac{\textrm{Cov}^2(m_{t-k},z_t)}{\textrm{Var}(m_t)\textrm{Var}(z_t)}=0$ indicates that $m_{t-k}$ and $z_t$ are unrelated, which means that $z_t$ does not include any information related to $m_{t-k}$ and, hence, the ESN cannot record $m_{t-k}$. In consequence, $M$ represents the total number of historical input data that each ESN can record.} 
The recurrent matrix $\boldsymbol{W}$ in each ESN model is given by: 
\begin{equation}
\boldsymbol{W}_l=\left[ {\begin{array}{*{20}{c}}
{{0}}&{{0}}& \cdots &{{w}}\\
{{w}}&0&0&0\\
0& \ddots &0&0\\
0&0&{{w}}&0
\end{array}} \right],
\end{equation}
 and the input weight matrix is given by ${\boldsymbol{W}^\textrm{in}} = {\left[ {w_1^\textrm{in}, \ldots ,w_{N_W}^\textrm{in}} \right]^{\rm T}}$. We also define a matrix that will be used to derive the memory capacity of the ESNs, which can be given by:
 \begin{equation}
 \boldsymbol{V}  = \left[ {\begin{array}{*{20}{c}}
{w_1^\textrm{in}}&{w_{N_W}^\textrm{in}}& \cdots &{w_2^\textrm{in}}\\
{w_2^\textrm{in}}&{w_1^\textrm{in}}& \cdots &{w_3^\textrm{in}}\\
 \vdots & \vdots & \cdots & \vdots \\
{w_{N_W}^\textrm{in}}&{w_{N_W-1}^\textrm{in}}& \cdots &{w_1^\textrm{in}}
\end{array}} \right].
\end{equation}
 Based on the above definitions, we can invoke our result from \cite[Theorem 2]{7880663} to derive the memory capacity of single ESN model, which can be given as follows.

\begin{corollary}[Single ESN model]\label{co:1}\emph{Given the recurrent matrix $\boldsymbol{W}$ and the input matrix $\boldsymbol{W}^\textrm{in}$ that guarantees the matrix $\boldsymbol{V}$ regular, the memory capacity of the single ESN model is:
\begin{equation}
M=N_W-1+w^{2N_W}.
\end{equation}
}
\end{corollary}
\begin{proof}{\color{black}Given the input stream vector $\boldsymbol{m}_{\ldots t}=\left[m_1,\ldots,m_{t-1},m_t\right]$, we can calculate the activations $\boldsymbol{\mu}_{j,t}$ using (\ref{eq:state}). The output weight matrix of the ESN model can be given by ${\boldsymbol{W}^\textrm{out}}={\boldsymbol{R}}^{-1}\boldsymbol{p}_{k}$, where ${\boldsymbol{R}}={ {\mathbb E}}\left[\boldsymbol{\mu}_t\left({\boldsymbol{\mu}}_t\right)^{\rm T}\right]$ represents the covariance matrix with $\boldsymbol{\mu}_t=\left[ \mu_{1,t},\ldots, \mu_{N_W,t}\right]$  and $\boldsymbol{p}_{k}={\mathbb E}\left[\boldsymbol{\mu}_tm_{t-k}\right]$. Assume that $\boldsymbol{w}^\textrm{in}_{N_W \ldots 1}=\left[ w_{N_W}^{\textrm{in}}, w_{N_W-1}^{\textrm{in}}, \ldots, w_{1}^{\textrm{in}} \right]$ and ${\mathop {\rm {rot}}}_k\left( \boldsymbol{w}^\textrm{in}_{N_W\ldots 1} \right)$ is an operator that rotates vector $\boldsymbol{w}^\textrm{in}_{N_W\ldots 1}$ by $k$ positions to the right. We have ${\boldsymbol{W}}^\textrm{out}=(1-w^{2N_W})w^{k}{\boldsymbol A}^{-1}{\mathop {\rm rot}}_{k}(\boldsymbol{w}_{1\ldots N_W}^\textrm{in})$, where ${\boldsymbol{A}}=\boldsymbol{V}^{\rm T}\Gamma ^2\boldsymbol{V} $ with $\Gamma=\diag \left(1,w, \ldots, w^{N_W-1}\right)$. Based on $\boldsymbol{W}^\textrm{out}$, we can the covariance of the output with the $k$-slot delayed input, which is given by $\textrm{Cov}(z_t,m_{t-k})= (1-w^{2N_W})w^{2k}\sigma ^2\zeta _k$. We can also obtain $\textrm{Var}(z_t)=\mathbb{E}\left[z_{t} z_{t} \right]=(1-w^{2N_W})w^{2k}\sigma ^2\zeta _k$. Since  $\textrm{Var}(m_t)=\sigma^2$, we have $M= \sum_{k=1}^{\infty}\frac{\textrm{Cov}^2(m_{t-k},z_t)}{\textrm{Var}(m_t)\textrm{Var}(z_t)}= N_W-1+w^{2N_W}$. }
\end{proof} 
From Corollary \ref{co:1}, we can see that the memory capacity of the single ESN model depends on the number of neurons and values of the recurrent matrix. Corollary \ref{co:1} also shows that the memory capacity of the single ESN model will not exceed $N_W$. That means the single ESN model based federated learning algorithm can only record $N_W$ locations or orientations.

Next, we derive the memory capacity of the parallel ESN model, which can be given by the following theorem.

\begin{theorem}[Parallel ESN] \label{th:1}
 \emph{ Given a parallel ESN model during which $L$ ESN models are parallel connected with each other, each ESN model's input weight matrix $\boldsymbol{W}^\textrm{in}$ that guarantees the matrix $\boldsymbol{V}$ regular and recurrent matrix $\boldsymbol{W}$, then the memory capacity of each parallel ESN can be given by:}
 \begin{equation}
  M=N_W-1+w^{2N_W}.
  \end{equation}
\end{theorem}
\begin{proof} See Appendix A.
\end{proof}
 Theorem \ref{th:1} shows that the memory capacity of a parallel ESN model is similar to the memory capacity of a single ESN. Hence, adding multiple ESN models will not increase the memory capacity. This is due to the fact that, in a parallel ESN model, there is no connection among the ESNs, as shown in Fig. \ref{fig7b}. Therefore, the input of the parallel ESN model will separately connect to each single ESN and, hence, the parallel ESN models do not need to use more neurons to record the input data compared to the single ESN model. Theorem \ref{th:1} also shows that the memory capacity of a parallel ESN depends on the number of neurons in each ESN model and the values of the recurrent weight matrix of each ESN model. Accordingly, we can increase the value of output weight matrix and the number of neurons in each ESN model to increase the memory capacity of the parallel ESN models. As the memory capacity of the parallel ESN models increases, BSs can record more users' data to predict the users' locations and orientations accurately.    
Next, we derive the memory capacity of the series ESN model.
\begin{theorem}[Series ESN model]\label{th:2}
\emph{ Given a series ESN model during which $L$ ESN models are series connected with each other, a recurrent matrix $\boldsymbol{W}$ of each ESN model, and each ESN model's input weight matrix $\boldsymbol{W}^\textrm{in}$ that guarantees the matrix $\boldsymbol{V}$ regular, the memory capacity of each series ESN model is:}
\begin{equation}
M =    {\left(1-w^{2N_W}\right)}^{L-1}\left(N_W-1+{w^{2N_W}}\right).
\end{equation}
\end{theorem}
\begin{proof} See Appendix B.
\end{proof}
From Theorem \ref{th:2}, we can see that  the memory capacity of each series ESN model is smaller than the memory capacity of a single ESN or a series ESN. Theorem \ref{th:2} also shows that the memory capacity of each series ESN model decreases as the number of ESN models $L$ increases. Thus, it would be better to use a single ESN model or a parallel ESN model to predict the users' locations and orientations.  

Theorems \ref{th:1} and \ref{th:2} derive the memory capacities of the parallel ESN model and the series ESN model with single input. Next, we formulate the memory capacity of a single ESN model given multiple inputs, which is given by the following theorem.

\begin{theorem}[Multi-input single ESN]\label{th:3}
\emph{  Consider a single ESN with a recurrent matrix $\boldsymbol{W}$, input vector $\boldsymbol{m}_t=\left[m_{1t},\ldots, m_{Kt}\right]$, the input weight matrix $\boldsymbol{W}^\textrm{in}$ that guarantees the matrix $\boldsymbol{V}$ regular, the memory capacity of each single ESN is
  \begin{equation}
    M =  \left(\frac{\sum_{l=1}^{K}\sigma _l^2}{\sum_{k=1}^{K}\sum_{n=1}^{K}\rho _{kn}\sigma _k\sigma _n}\right)^2\left(N_W-1+w^{2N_W}\right),
  \end{equation}
  where $\rho _{kn}$ represents the correlation coefficient between input $m_{kt}$ and $m_{nt}$.
  }
  
\end{theorem}
\begin{proof} See Appendix C.
\end{proof}

From Theorem \ref{th:3}, we can observe that the correlation among input elements in vector $\boldsymbol{m}_t$ will affect the memory capacity of each ESN model. In particular, as the correlation of the input data increases, the memory capacity of the ESN model increases. This is because the ESN can use more input data to predict the users' locations and orientations, hence improving the predictions accuracy. Therefore, it would be better to jointly predict the users' locations and orientations.

 Theorems \ref{th:1}-\ref{th:3} allow each BS to determine its ESN model, the number of neurons $N_W$ in each ESN model, and the values of the recurrent matrix $\boldsymbol{W}$ as the size of the data collected by each BS changes. A parallel ESN model has a larger memory capacity compared with the series ESN model and is more stable than the single ESN model, and, hence, a parallel ESN model can record more historical data to predict the users' orientations and locations so as to improve the prediction accuracy. As the prediction accuracy is improved, the BSs can determine the user association more accurately. Hence, the BIP of the users can be minimized. Therefore, we use the parallel ESN model in our proposed algorithm. 

%

 \section{User Association for VR Users} 
Based on the analysis presented in Sections \ref{se:esn} and \ref{se:mc}, each BS can predict the users' locations and orientations. Next, we explain how to use these predictions to find the user association for each VR user. Given the predictions of the locations and orientations, the BIP minimization problem in (\ref{eq:max}) can be rewritten as follows: 
\begin{equation}\label{eq:finalmax}
\mathop {\min }\limits_{\boldsymbol{a}_{i,t}^\textrm{UL}, \boldsymbol{a}_{i,t}^\textrm{DL}}\sum\limits_{i \in \mathcal{U} }P_{i}\left(\hat x_{it}, \hat y_{it}, G_A, \hat \chi_{it}, \boldsymbol{a}_{i,t}^\textrm{UL}, \boldsymbol{a}_{i,t}^\textrm{DL}\right).
\end{equation}
We use the reinforcement learning algorithm given in \cite{5700414} to find a sub-optimal solution of the problem in (\ref{eq:finalmax}). {\color{black}In the reinforcement learning algorithm given in \cite{5700414}, the actions are the user association schemes, the states are the strategies of other BSs, and the output is the estimated BIP. Hence, this reinforcement learning algorithm can learn the VR users state and exploit different actions to adapt the user association according to the predictions of the users' locations and orientations.} After the learning step, each BS will find a sub-optimal user association to service the VR users.
{\color{black}To simplify the learning process and improve the convergence speed, we first select the uplink user association scheme. This is because as the uplink user association is determined, the BSs that the users can associate in downlink will be determined, as follows:}

 \begin{proposition}\label{th:4}
\emph{Given the predicted location and orientation of user $i$ at time $t$ as well as the uplink user association ${a}_{i*,t}^\textrm{UL}$, the downlink cell association for a VR user $i$ is:
\begin{equation}\label{eq:content}
 {a}_{ik,t}^\textrm{DL}= \mathds{1}_{\left\{ \frac{D\left( \boldsymbol{l}_{i,t}\left({a}_{ik,t}^\textrm{DL}c_{ik}^\textrm{DL}\left(\hat x_{it}, \hat y_{it}, b_i\left(\hat {\chi}_{it}\right), n_{ik}\right)\right)\right)}{{a}_{ik,t}^\textrm{DL}c_{ik}^\textrm{DL}\left(\hat x_{it}, \hat y_{it}, b_i\left(\hat {\chi}_{it}\right), n_{ik}\right)}  \leqslant \gamma_\textrm{D}-\frac{A}{{a}_{i*,t}^\textrm{UL}c_{i*}^\textrm{UL}\left(\hat x_{it}, \hat y_{it}\right)}\right\}}  \wedge  \mathds{1}_{\left\{ \boldsymbol{l}_{i,t}\left({a}_{ik,t}^\textrm{DL}c_{ik}^\textrm{DL}\left( \hat x_{it}, \hat y_{it}, b_i\left(\hat {\chi}_{it}\right), n_{ik}\right)\right)\boldsymbol{m}_{i,t}\left( G_A\right)  \geqslant \gamma_\textrm{Q} \right\}},
\end{equation}
where $ {a}_{ik,t}^\textrm{DL}$ is the downlink user association obtained in (\ref{eq:content}). $c_{i*}^\textrm{UL}\left(x_{it},y_{it}\right)$ is the uplink data rate of user $i$.
}
\end{proposition} 
\begin{proof} For downlink user association, each VR user $i$ needs to find a BS that can guarantee the transmission delay and VR video quality. Since we have determined the user association over uplink, the maximum time used for VR video transmission can be given by $\gamma_\textrm{D}-\frac{A}{{a}_{i*,t}^\textrm{UL}c_{i*}^\textrm{UL}\left(\hat x_{it}, \hat y_{it}\right)}$. Consequently, user $i$ needs to connect with a BS that can satisfy the transmission delay requirement of user $i$, i.e., $\frac{D\left( \boldsymbol{l}_{i,t}\left({a}_{ik,t}^\textrm{DL}c_{ik}^\textrm{DL}\left(\hat x_{it},\hat y_{it}, b_i\left(\hat{\chi}_{it}\right), n_{ik}\right)\right)\right)}{{a}_{ik,t}^\textrm{DL}c_{ik}^\textrm{DL}\left(\hat x_{it},\hat y_{it}, b_i\left(\hat{\chi}_{it}\right), n_{ik}\right)}  \leqslant \gamma_\textrm{D}-\frac{A}{{a}_{i*,t}^\textrm{UL}c_{i*}^\textrm{UL}\left(\hat x_{it},\hat y_{it}\right)}$. Moreover, user $i$ needs to associate with a BS that can meet the requirement of VR video quality, i.e., $\boldsymbol{l}_{i,t}\left({a}_{ik,t}^\textrm{DL}c_{ik}^\textrm{DL}\left( \hat x_{it}, \hat y_{it}, b_i\left(\hat {\chi}_{it}\right), n_{ik}\right)\right)\boldsymbol{m}_{i,t}\left( G_A\right)  \geqslant \gamma_\textrm{Q}$. Thus, if BS $k$ can satisfy the conditions: $\frac{D\left( \boldsymbol{l}_{i,t}\left({a}_{ik,t}^\textrm{DL}c_{ik}^\textrm{DL}\left(\hat x_{it},\hat y_{it}, b_i\left(\hat{\chi}_{it}\right), n_{ik}\right)\right)\right)}{{a}_{ik,t}^\textrm{DL}c_{ik}^\textrm{DL}\left(\hat x_{it},\hat y_{it}, b_i\left(\hat{\chi}_{it}\right), n_{ik}\right)}  \leqslant \gamma_\textrm{D}-\frac{A}{{a}_{i*,t}^\textrm{UL}c_{i*}^\textrm{UL}\left(\hat x_{it},\hat y_{it}\right)}$ and $\small \boldsymbol{l}_{i,t}\left({a}_{ik,t}^\textrm{DL}c_{ik}^\textrm{DL}\left( \hat x_{it}, \hat y_{it}, b_i\left(\hat {\chi}_{it}\right), n_{ik}\right)\right)\boldsymbol{m}_{i,t}\left( G_A\right)  \geqslant \gamma_\textrm{Q}$, user $i$ can associate with it. This completes the proof.    
\end{proof}
From Proposition \ref{th:4}, we can see that the user association of each user $i$ depends on user $i$'s location and orientation.
Proposition \ref{th:4} shows that, for each user $i$, the uplink user association will affect the downlink user association. This is due to the fact that the VR system has determined the total transmission delay of each user. As a result, when the uplink user association is determined, the uplink transmission delay and the requirement of the downlink transmission delay will be determined. 

   \section{Simulation Results and Analysis}

\begin{table}
  \newcommand{\tabincell}[2]{\begin{tabular}{@{}#1@{}}#2\end{tabular}}
\renewcommand\arraystretch{0.7}
 \caption{
    \vspace*{-0.3em} \scriptsize SYSTEM PARAMETERS}\label{ta:parameters}\vspace*{-1em}
\centering  
\begin{tabular}{|c|c|c|c|c|c|}
\hline
\textbf{Parameter}& \textbf{Value} &\textbf{Parameter} & \textbf{Value} &\textbf{Parameter}& \textbf{Value} \\
\hline
$P_B $ & 30 dBm & $d_0$ & 5 m&$Y$&$10$  \\
\hline
$P_U$ & 10 dBm & $f_c$ & 28 GHz & $T$ & 5 \\
\hline
$\sigma$ & -94 dBm& $c$ & $3\times 10^8$ m/s &$\gamma$&0.5  \\
\hline
 $F^\textrm{UL}$ & 10 Mbit& ${\color{black}\varpi_\textrm{LoS}}, {\color{black}\varpi_\textrm{NLoS}}$ & 2, 2.4 &$\lambda$ & 0.005\\
\hline
$F^\textrm{DL}$ & 10 Mbit & ${\mu_{\sigma_\textrm{LoS}}}$, ${\mu_{\sigma_\textrm{NLoS}}}$ & 5.3, 5.27&$w$ &0.98\\
\hline
$N_W $ & 30 &$G_A$&11& $L$ & 3\\
\hline
 $\beta$  & 2 &$\gamma_{D}$& 10 ms& $V$ & 10 \\
\hline
$M$ & 15 dB & $\gamma_Q$ & 0.8 & $\vartheta  $ & 2 \\
\hline
$m$&0.7 dB& ${\sigma_i}^2$&0.193&${\sigma_{G_A|i}^2}$ & 0.151  \\
\hline
$\phi $ & $30^\circ$&$A$&50 kbits &${\sigma_B}^2$ & 0.05  \\
\hline
\end{tabular}
 \vspace{-0.7cm}
\end{table}

  For our simulations, we consider a circular area with radius $r = 500$ m, $U=20$ wireless VR users, and $B=5$ BSs distributed uniformly. {\color{black} To simulate blockage, each user is considered as a two-dimensional point. For simplicity, we ignore the altitudes of the BSs and the height of the users. If blockage points are located between a user and its BS, the communication link will be considered to be NLoS.} Real data traces for locations are collected from 50 students at the Beijing University of Posts and Telecommunications. The locations of each student is collected every hour during 9:00 am -- 9:00 pm. {\color{black} For orientation data collection, we searched 25 videos related to a first-person shooter game from youTube. Then, we input these VR videos to HTC Vive devices. The HTC Vive deveploper system can directly measure the movement of the VR videos using HTC Vive devices. We arbitrarily combine one user's locations with one orientation for each VR user.} In simulations, a parallel ESN model{\footnote{\color{black}The code can be found in https://github.com/lasisal/deepESN.}} is used for the proposed algorithm due to its stability and large memory capacity. The other system parameters are listed in Table \ref{ta:parameters}. 
  For comparison purposes, we consider the deep learning algorithm in \cite{6503981} and the ESN algorithm in \cite{chen2016caching}, as two baseline schemes. {\color{black}The deep learning algorithm in \cite{6503981} is a deep autoencoder that consists of multiple layers of restricted Boltznann machines. The centralized ESN-based learning algorithm in \cite{chen2016caching} is essentially a single layer ESN algorithm. The input and output of the centralized ESN and deep learning algorithms are similar to the proposed algorithm. However, for the deep learning algorithm and the centralized ESN algorithm, each BS can use only its collected data to train the learning model. Both the centralized and deep learning algorithms are trained in an offline manner.} All statistical results are averaged over a large number of independent runs.

\begin{figure}[!t]
\centering
\vspace{0cm}
\subfigure[]{
\label{fig3a} 
\includegraphics[width=5.2cm]{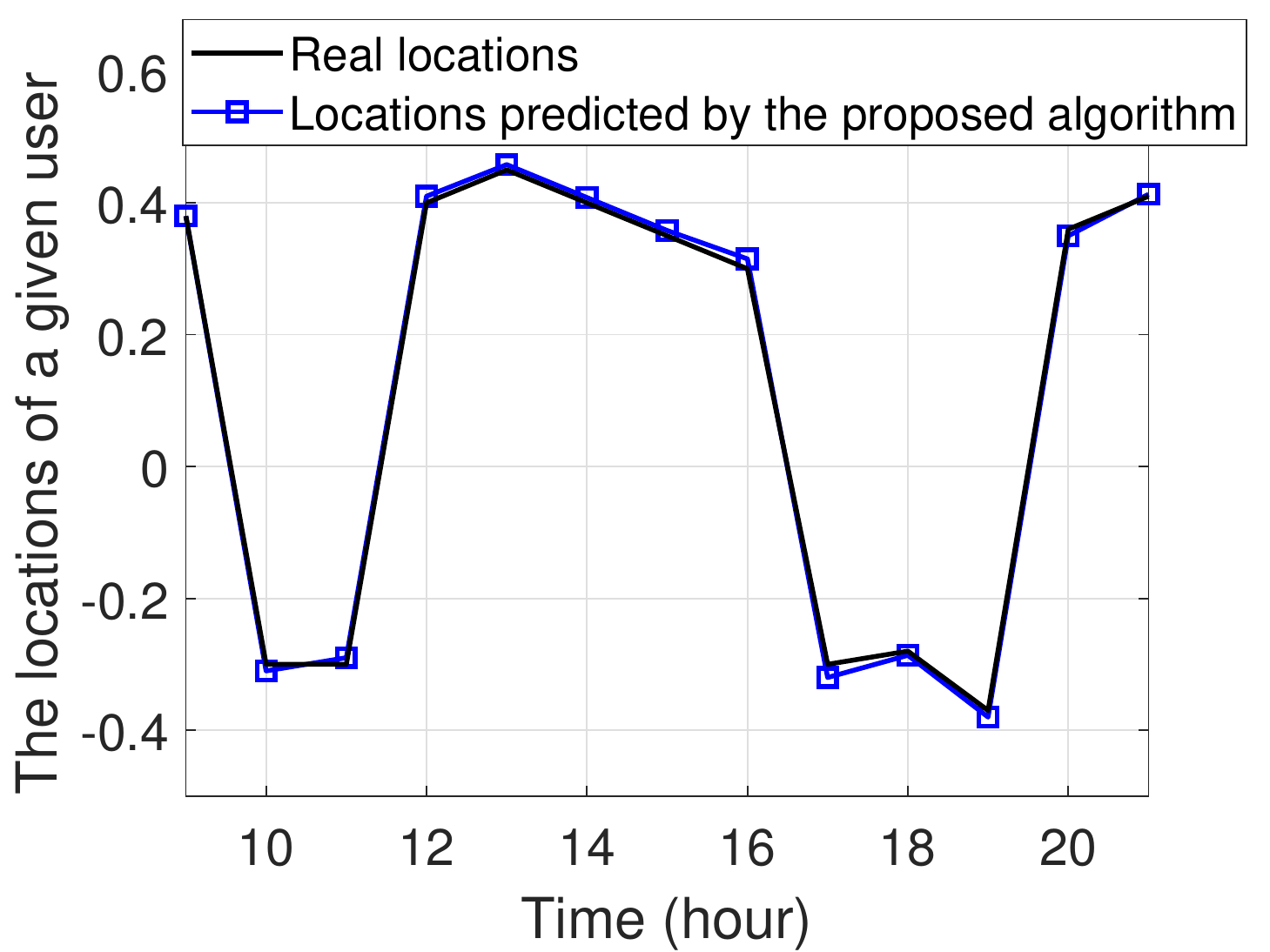}}
\subfigure[]{ 
\label{fig3b} 
\includegraphics[width=5.2cm]{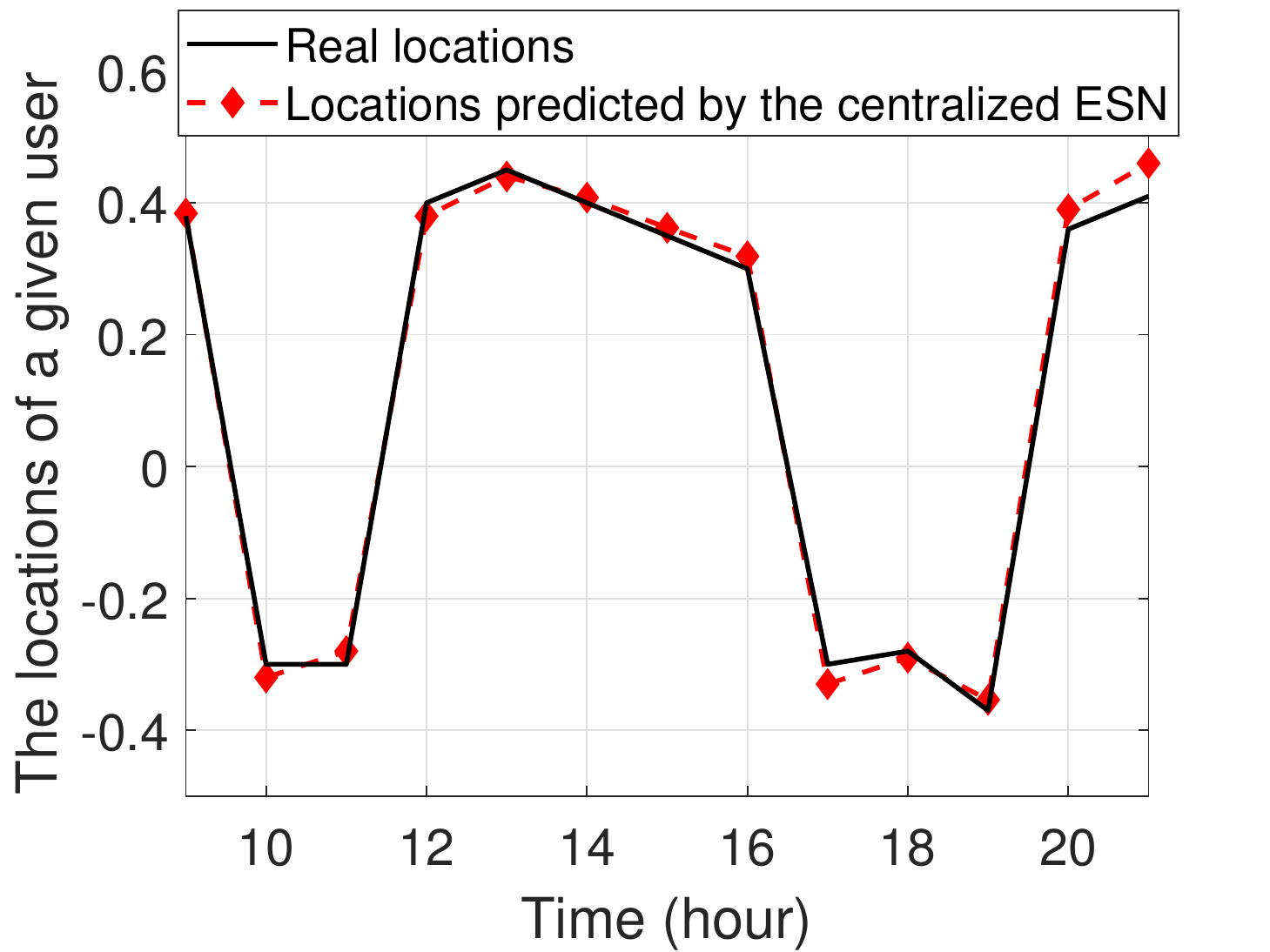}}
\subfigure[]{ 
\label{fig3b} 
\includegraphics[width=5.2cm]{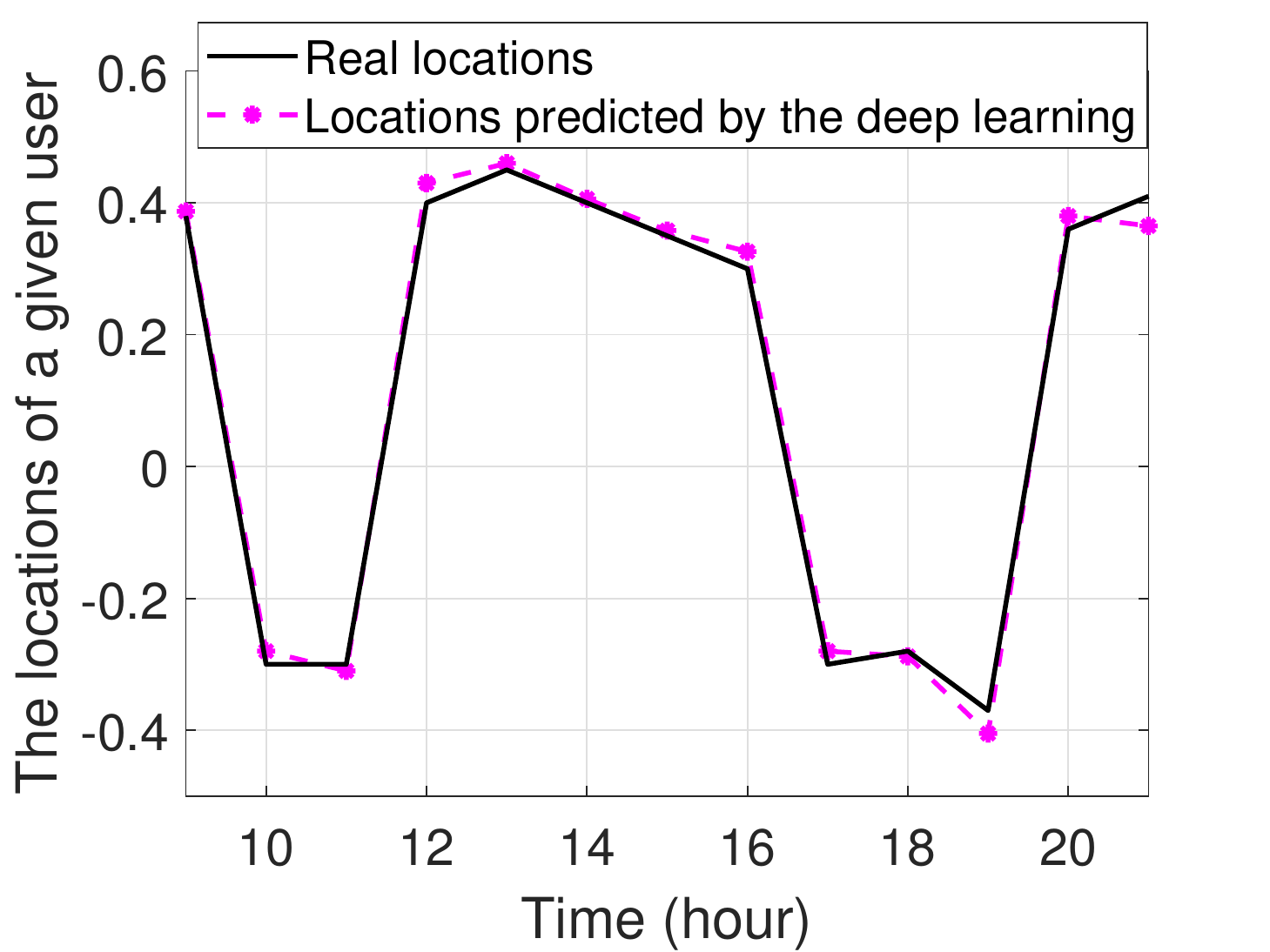}}
  \vspace{-0.4cm}
 \caption{\label{figure3} Predictions of the VR users' orientations and locations as time elapses.}
 \vspace{-0.1cm}
\end{figure}

\begin{figure}[!t]
\centering
\vspace{0cm}
\subfigure[]{
\label{fig4a} 
\includegraphics[width=5.2cm]{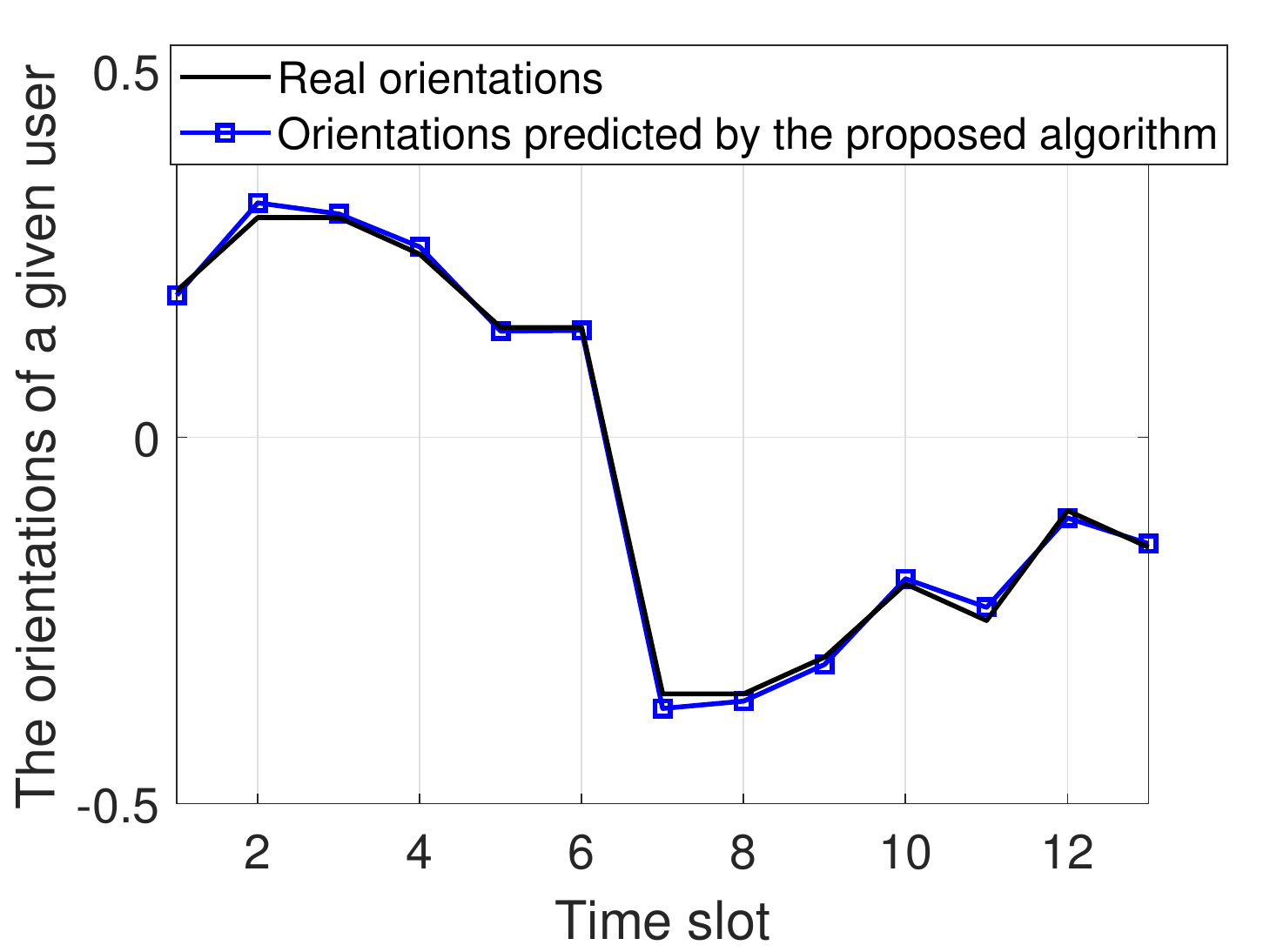}}
\subfigure[]{ 
\label{fig4b} 
\includegraphics[width=5.2cm]{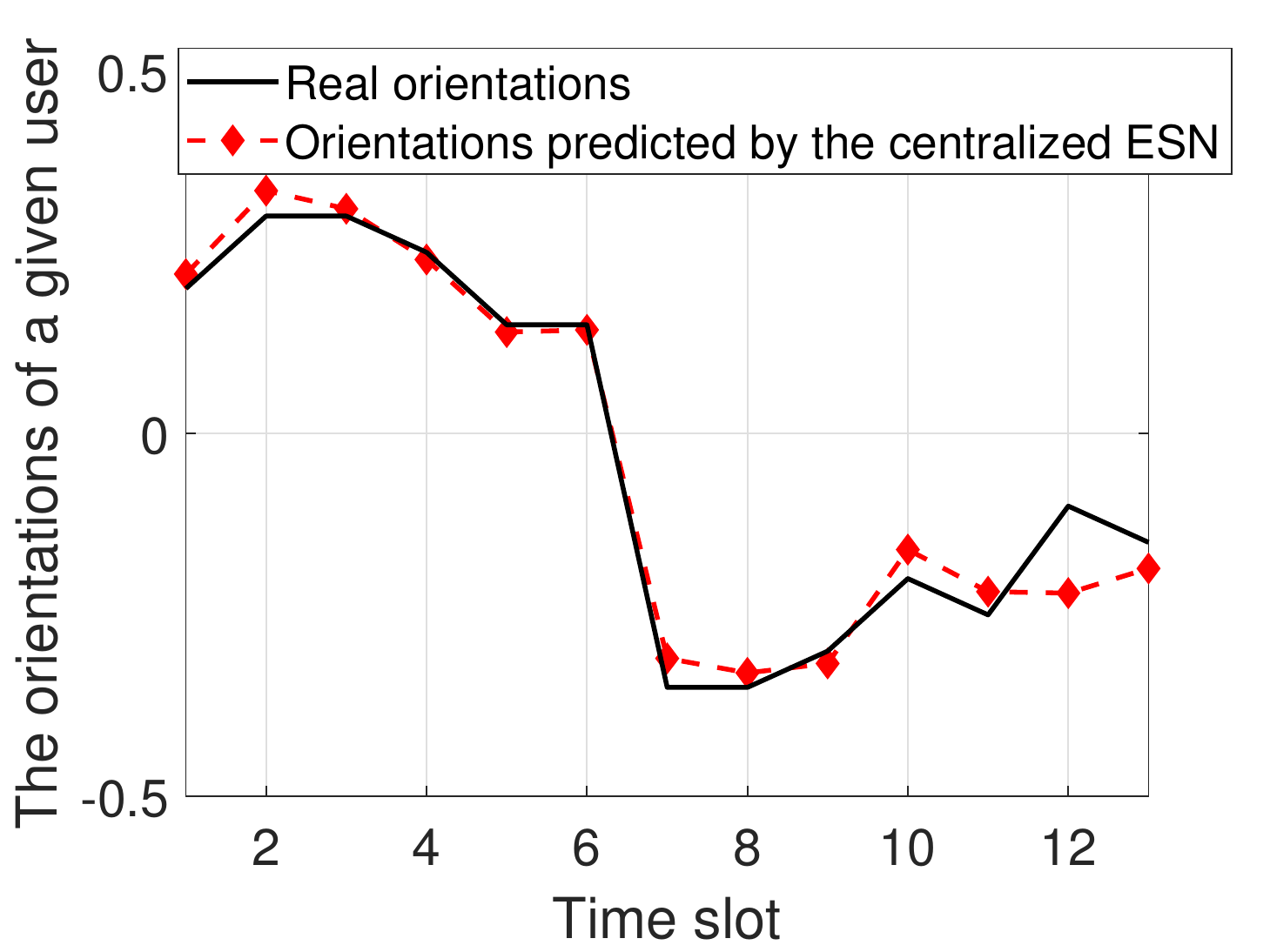}}
\subfigure[]{ 
\label{fig4c} 
\includegraphics[width=5.2cm]{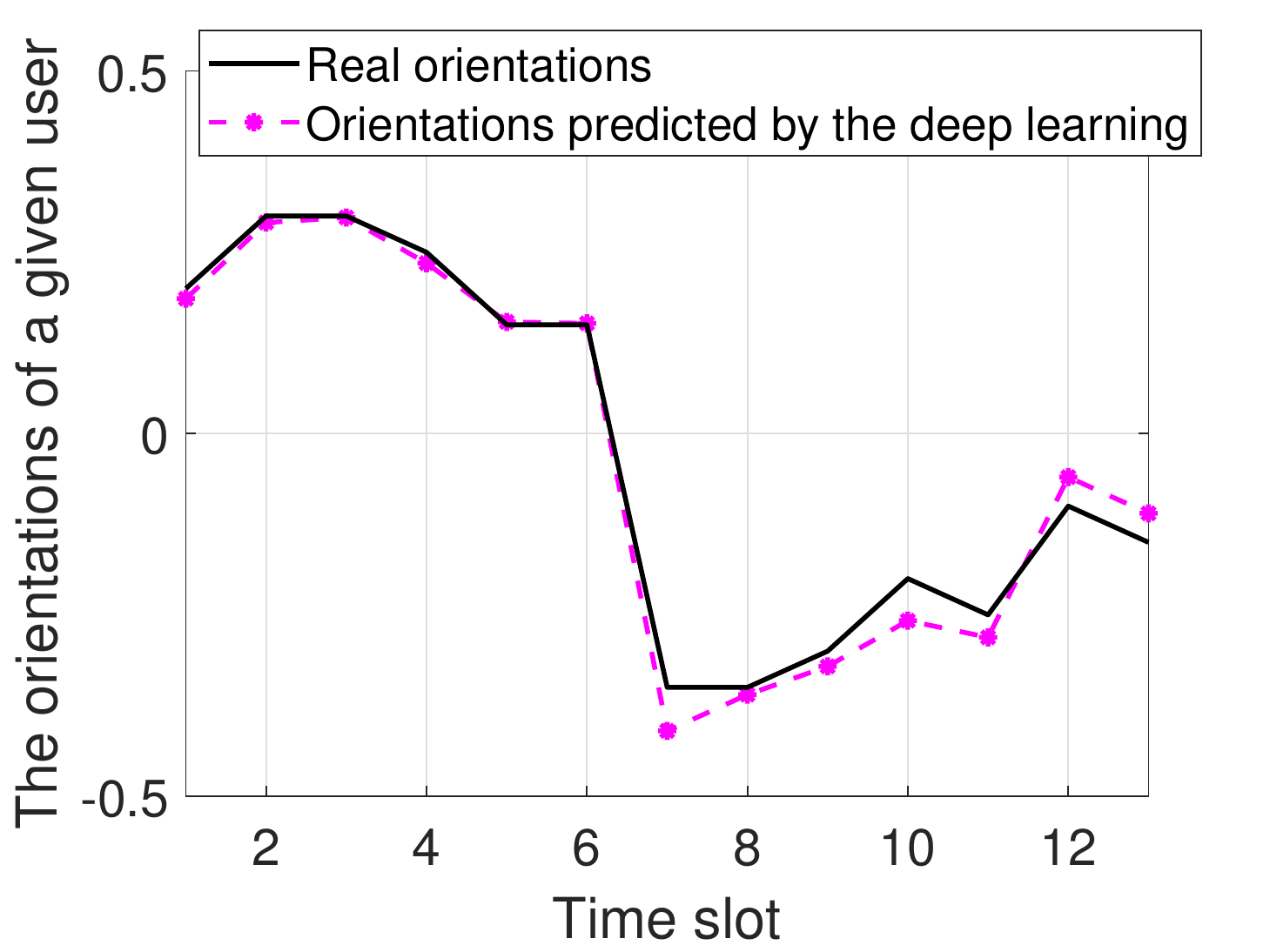}}
  \vspace{-0.4cm}
 \caption{\label{figure4} Predictions of the VR users' orientations and locations as time elapses.}
 \vspace{-0.8cm}
\end{figure}

Figs. \ref{figure3} and \ref{figure4} show the predictions of the VR users' locations and orientations as time elapses. To simplify the model training, the collected data related to locations and orientations are mapped to $\left[-0.5,0.5\right]$. {\color{black}The orientation and location of each user are, respectively, mapped by the function $\frac{\chi_{it}}{360^\circ}-0.5$ and $\frac{z}{z_{\max}}-0.5$ where $z=\left(\sum\limits_{n = 1}^{\hat x_{it}+\hat y_{it}} n\right)\times\hat y_{it}$.} From Figs. \ref{figure3} and \ref{figure4}, we observe that the proposed algorithm can predict the users' locations and orientations more accurately than the centralized ESN and deep learning algorithms. Figs. \ref{fig4b} and \ref{fig4c} also show that {\color{black}the prediction error mainly occurs} at time slot 8 to 12. This is due to the fact that the proposed algorithm can build a learning model that predicts the entire locations and orientations of each user. In particular, the output weight matrices of all ESN algorithms implemented by each BS will converge to a common matrix. Hence, BSs can predict the entire locations and orientations of each VR user.

\begin{figure}[!t]
\centering
\vspace{0cm}
\subfigure[]{
\label{figure1a} 
\includegraphics[width=8cm]{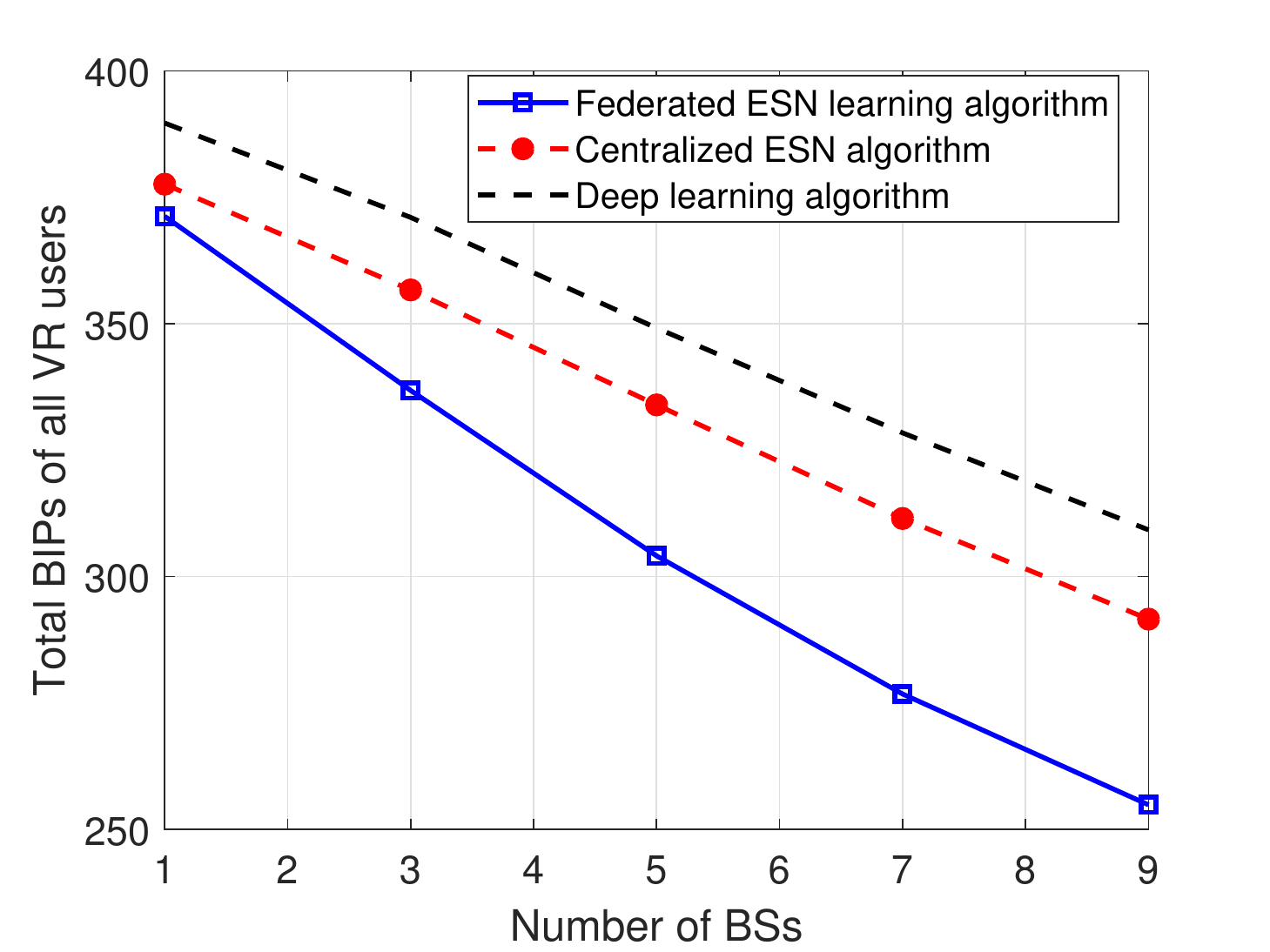}}
\subfigure[]{ 
\label{figure1b} 
\includegraphics[width=8cm]{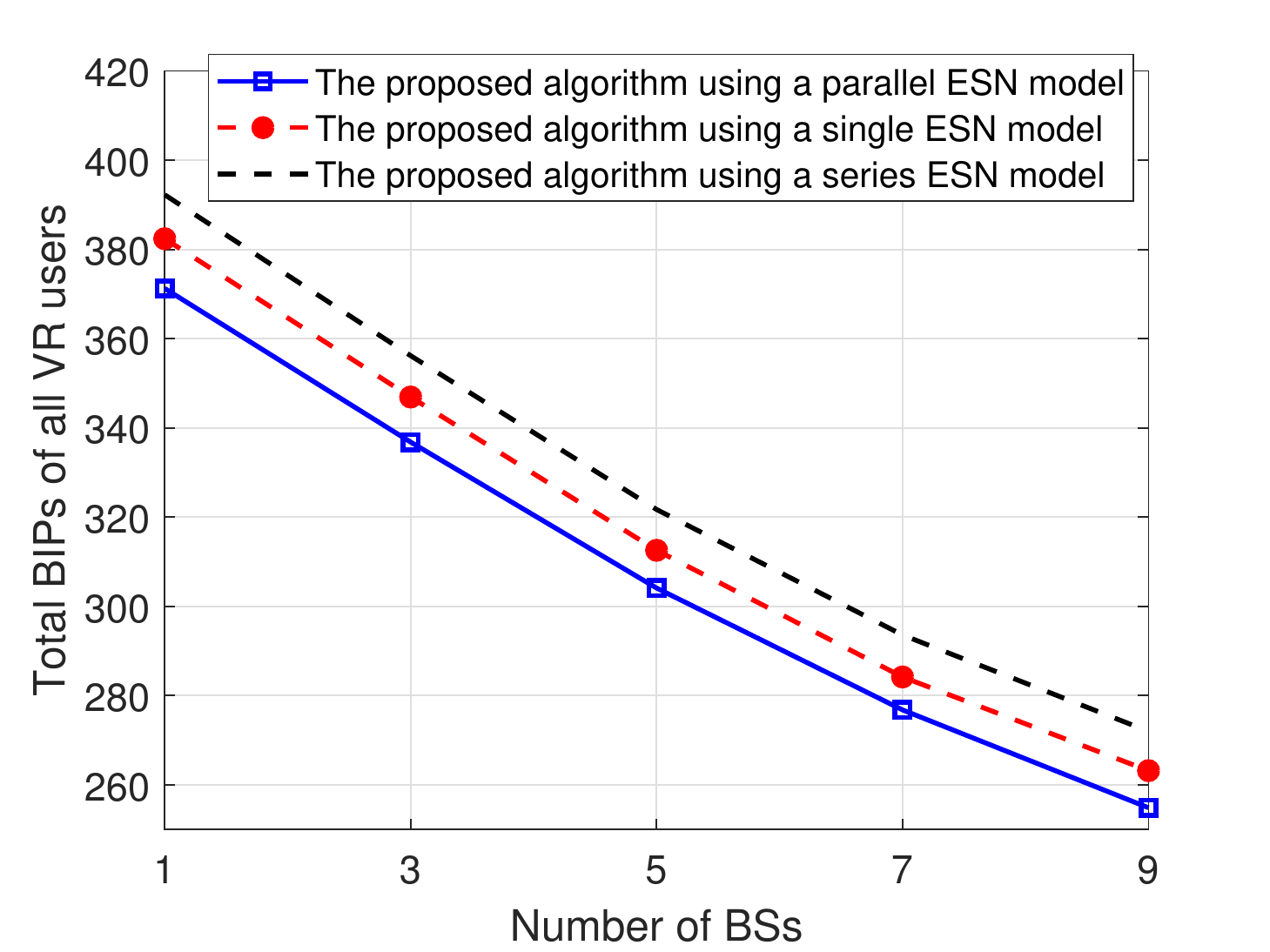}}
  \vspace{-0.8cm}
 \caption{\label{figure1}Total BIP experienced by VR users as the number of BSs varies.}
 \vspace{-0.8cm}
\end{figure}

%
Fig. \ref{figure1} shows how the total BIP of all VR users changes as the number of BSs varies. From Fig. \ref{figure1}, we can see that, as the number of BSs increases, the total BIP of all VR users decreases. That is because as the number of BSs increases, the VR users have more connection options. Hence, the blockage caused by human bodies will be less severe, thereby improving the data rates of VR users. Fig. \ref{figure1a} also shows that the proposed algorithm can achieve up to 16\% and 26\% reduction in the number of BIP, respectively, compared to centralized ESN algorithm and deep learning algorithm for a network with 9 BSs. These gains stem from the fact that the centralized ESN and deep learning algorithms can partially predict the locations and orientation of each VR user as they rely only on the local data collected by a BS. In contrast, the proposed algorithm facilitates cooperation among BSs to build a learning model that can predict the entire users' locations and orientations. Fig. \ref{figure1b} shows that the proposed algorithm using a parallel ESN model can achieve up to 8\% and 14\% gains in terms of the total BIP of all users compared to the proposed algorithm with a single ESN model and with a series model. Clearly, compared to a single ESN, using a parallel ESN model can increase the stability of the proposed algorithm. Meanwhile, the memory capacity of a parallel model is larger than a series ESN model thus improving the prediction accuracy and reducing BIP for users.

\begin{figure}[!t]
  \begin{center}
   \vspace{0cm}
    \includegraphics[width=9cm]{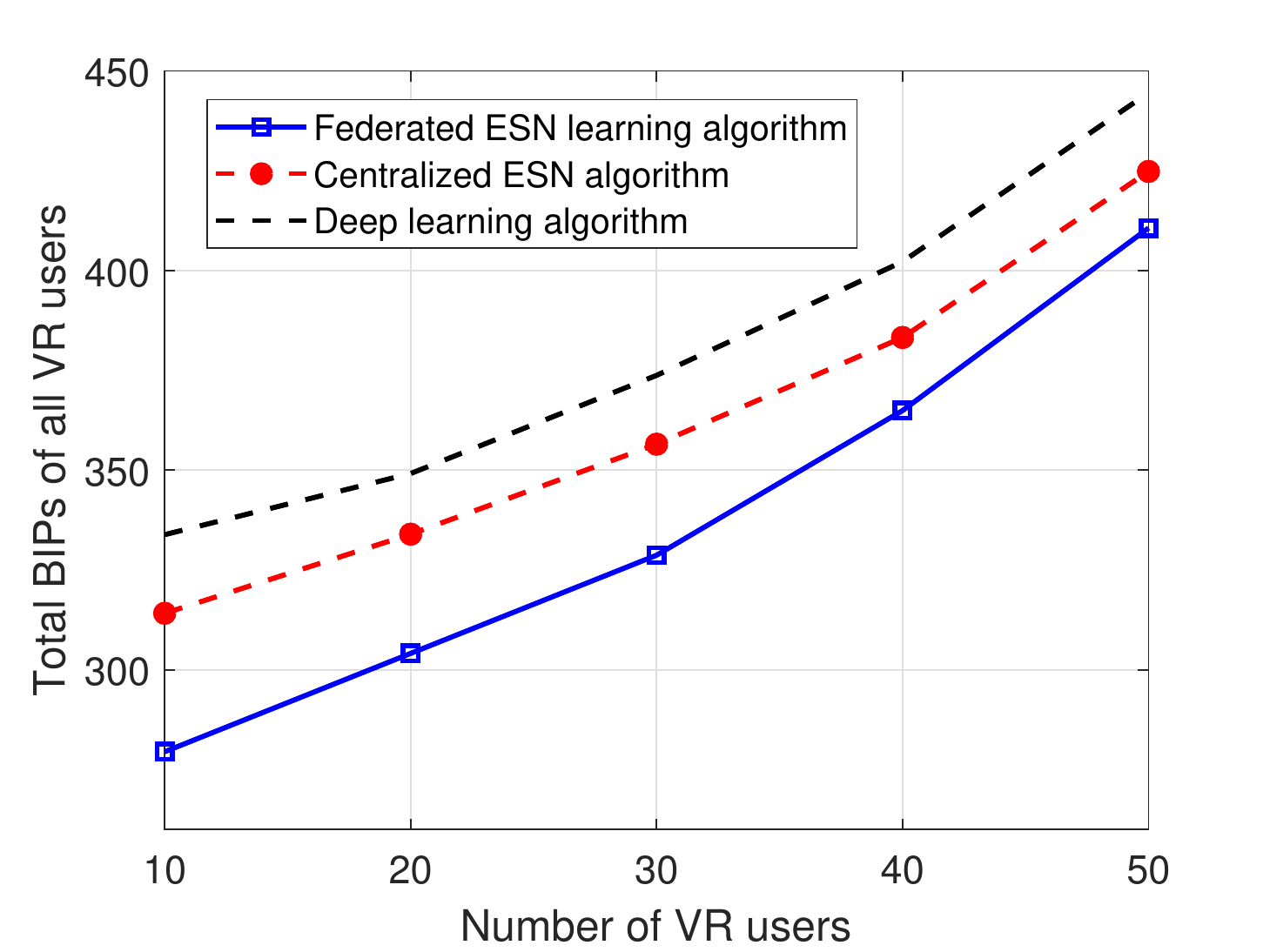}
    \vspace{-0.3cm}
    \caption{\label{figure7} Total BIP of all VR users as the number of VR users varies.}
  \end{center}\vspace{-0.8cm}
\end{figure} 

In Fig. \ref{figure7}, we show how the total BIP of all VR users changes with the number of VR users. This figure shows that, with more VR users, the total BIP of all VR users increases rapidly due to an increase in the uplink delay, as the sub-6 GHz bandwidth is shared by more users. Fig. \ref{figure7} also shows that the gap between the proposed algorithm and the centralized ESN algorithm decreases as more VR users are present in the network.. Clearly, with more VR users, it becomes more probable that a user located between a given VR user and its associated BS blocks the mmWave link. Thus, as the number of users increases, more VR users will receive their VR videos over NLoS links and, the total BIP significantly increases. 

\begin{figure}[!t]
  \begin{center}
   \vspace{0cm}
    \includegraphics[width=9cm]{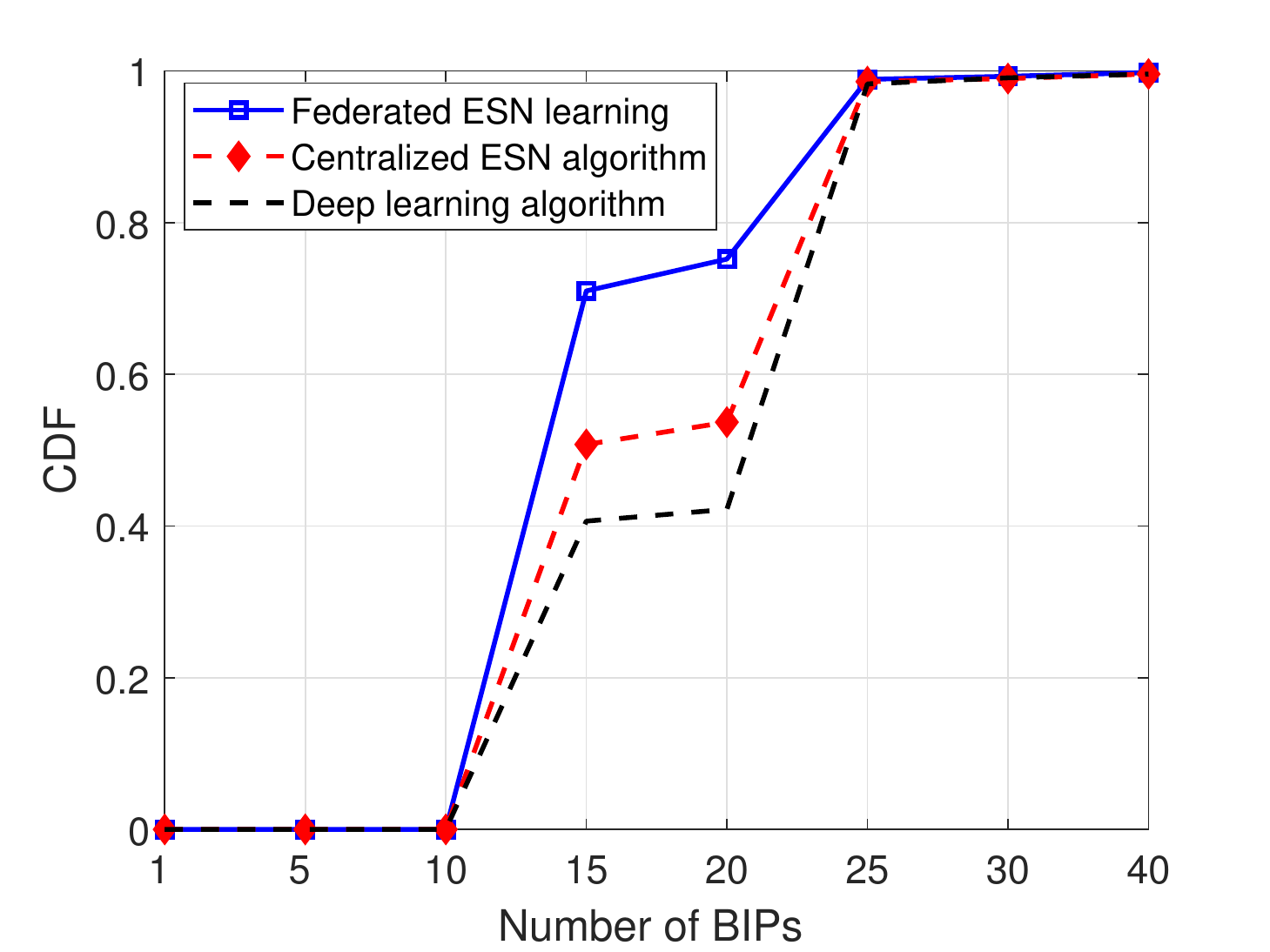}
    \vspace{-0.3cm}
    \caption{\label{figure2} CDFs of the BIP resulting from the different algorithms.}
  \end{center}\vspace{-1cm}
\end{figure}

In Fig. \ref{figure2}, we show the CDF for the VR users' BIP for all three algorithms. Fig. \ref{figure2} shows that the BIP of almost 98\% of users resulting from the considered algorithms will be larger than 10. This is due to the fact that the BIP will also be caused by other factors such as VR applications and user's awareness. In Fig. \ref{figure2}, we can also see that the proposed algorithm improves the CDF of up to {38\%} and {71\%} gains at a BIP of 25 compared to the centralized ESN and deep learning algorithms, respectively. These gains stem from the fact the ESNs are effective at analyzing the time related location and orientation data and, hence, they can accurately predict the users' locations and orientations.  

%

\begin{figure}[!t]
  \begin{center}
   \vspace{0cm}
    \includegraphics[width=9cm]{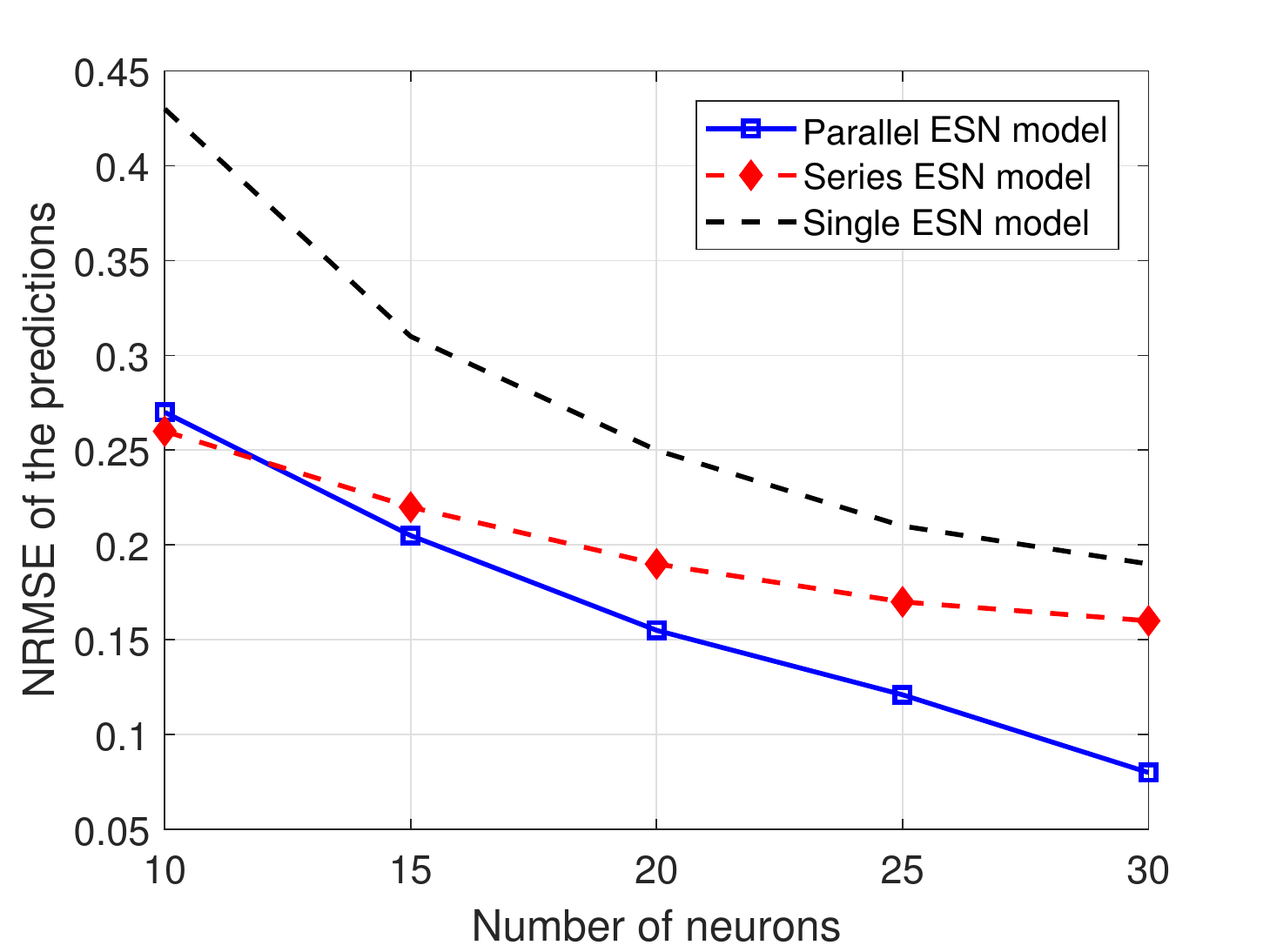}
    \vspace{-0.3cm}
   {\color{black} \caption{\label{figure6} Normalized root mean square error of the predictions as the number of neurons varies.}}
  \end{center}\vspace{-1cm}
\end{figure} 
Fig. \ref{figure6} shows how the normalized root mean square error (NRMSE) of the predictions changes as the number of neurons in each ESN model varies. In this figure, the NRMSE of the predictions is given by $\left\|\hat{\boldsymbol{y}}_{ij,t}-{\boldsymbol{e}}_{ij,t}  \right\|$. From Fig. \ref{figure6}, we can see that, with more neurons, the NRMSE of the predictions resulting from all of the considered ESN models decreases. This is because,  as the number of neurons increases, each ESN model can record more historical data related to the users' locations and orientations. Fig. \ref{figure6} also shows that the parallel model can achieve up to 37.5\% and 90\% gains in terms of NRMSE compared to the series model for the ESN models have 30 neurons. This is due to the fact that the prediction errors of a parallel ESN model is averaged over multiple outputs, thus, improving the prediction accuracy.

 \begin{figure}[!t]
  \begin{center}
   \vspace{0cm}
    \includegraphics[width=9cm]{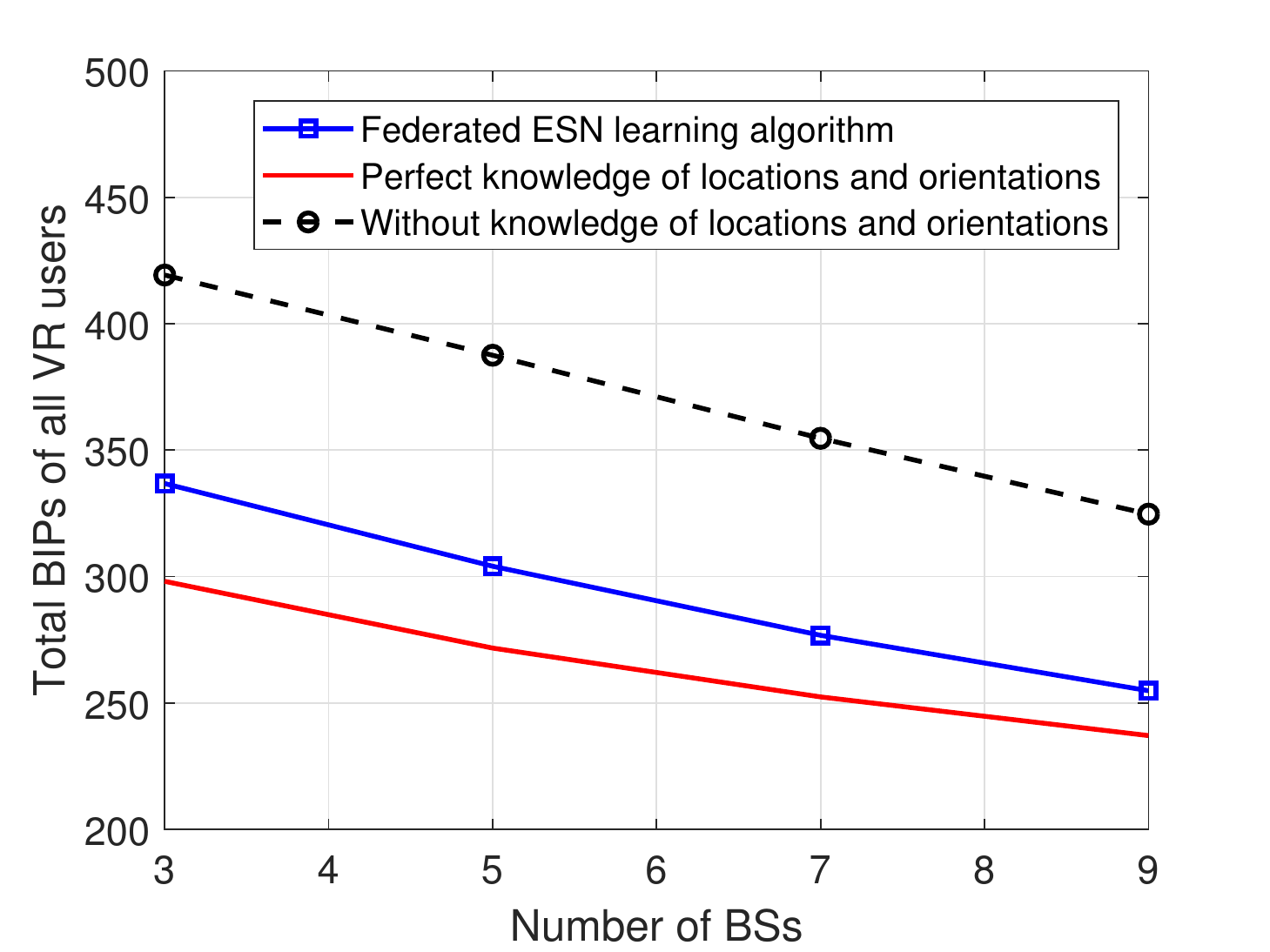}
    \vspace{-0.3cm}
    {\color{black}\caption{\label{figure11} Total BIP of all VR users as the number of BSs varies.}}
  \end{center}\vspace{-0.8cm}
\end{figure} 
 
{\color{black}Fig. \ref{figure11} shows how the total BIP of all VR users changes as the number of BSs varies. In this figure, all of the considered algorithms used the reinforcement learning algorithm given in [31] to solve the problem in (25). In the algorithm without the knowledge of the locations and orientations, the users are randomly associated with the BSs.  From Fig. \ref{figure11}, we can see that the proposed algorithm yields 23\% fewer BIP compared to the algorithm without the knowledge of the locations and orientations. This is due to the fact that the proposed algorithm uses federated ESN algorithm to predict the users' orientations and locations to optimize the user association and reduce the BIP of the users. From Fig. \ref{figure11}, we can also see that, a gap exists between the proposed algorithm and  the algorithm with the perfect knowledge of the users' orientations and locations. This gap stems from the prediction inaccuracy caused by the proposed algorithm.}
 

   \section{Conclusion}
  In this paper, we have developed a novel framework for minimizing BIP within VR applications that operate over wireless networks.
   To this end, we have developed a BIP model that jointly considers the VR applications, transmission delay, VR video quality, and the user's awareness. 
   We have then formulated an optimization problem that seeks to minimize the BIP of VR users by predicting users' locations and orientations, as well as determining the user association. To solve this problem, we have developed a novel federated learning algorithm based on echo state networks. The proposed federated ESN algorithm enables the BSs to train their ESN with their locally collected data and share these models to build a global learning model that can predict the entire locations and orientations of each VR user. To improve the prediction accuracy of the proposed algorithm, we derive a closed-form expression of the memory capacity for ESNs to determine the number of neurons in each ESN model and the values of the recurrent weight matrix. Using these predictions, each BS can determine the user association in both uplink and downlink. Simulation results have shown that, when compared to the centralized ESN and deep learning algorithms, the proposed approach achieves significant performance gains of BIP.

  \section*{Appendix}
\subsection{Proof of Theorem \ref{th:1}}\label{Ap:a}  
Given the input stream vector $\boldsymbol{m}_{\ldots t}=\left[m_1,\ldots,m_{t-1},m_t\right]$, the activations $\boldsymbol{\mu}_{j,t}$ in (\ref{eq:state}) of the reservoir neuron in ESN model $l$ at time $t$ can be given by:
  \begin{equation}\label{eq:mu}\small
    \begin{aligned}
      \mu_{j1,t}^{\left(l\right)}&=w_{1,l}^{\textrm{in}}m_t+ww_{N_W,l}^{\textrm{in}}m_{t-1}+w^2w_{N_W-1,l}^{\textrm{in}}m_{t-2}+\cdots +w^{N_W-1}w_{2,l}^{\textrm{in}}m_{t-(N_W-1)}\\
                  &+w^{N_W} w_{1,l}^{\textrm{in}} m_{t-N_W}+w^{N_W+1}w_{N_W,l}^{\textrm{in}}m_{t-(N_W+1)}+\cdots+w^{2N_W-1}w_{2,l}^{\textrm{in}}m_{t-(2N_W-1)}+w^{2N_W}w_{1,l}^{\textrm{in}}m_{t-2N_W}\\
                  &+w^{2N_W+1}w_{N_W,l}^{\textrm{in}}m_{t-(2N_W+1)}+\cdots,
    \end{aligned}
    \vspace{-0.5cm}
  \end{equation}
    \begin{equation}\nonumber
     \vdots 
   \end{equation} 
  \begin{displaymath}\small
    \begin{aligned}
 \!\!\!\!\! \!\!\!\!\! \!\!\!\!\! \!\!\!\!\! \!\!\!\!\! \!\!\!\!\! \!\!\!\!\! \!\!\!\!\! \!\!\!\!\! \!\!\!\!\! \!\!\!\!\! \!\!\!\!\! \!\!\!\!\! \!\!\!\!\! \!\!\!\!\! \!\!\!\!\! \!\!\!\!\!\mu_{jN_W,t}^{(l)}&=  w_{N_W,l}^{\textrm{in}}m_{t}+ww_{N_W-1,l}^{\textrm{in}}m_{t-1}+w^2w_{N_W-2,l}^{\textrm{in}}m_{t-2}+\cdots +w^{N_W-1}w_{1,l}^{\textrm{in}}m_{t-(N_W-1)}\;\;\;\;\;\;\;\;\;\;\;\;\;\;\;\\
                  &+w^{N_W}w_{N_W,l}^{\textrm{in}}m_{t-N_W}+w^{N_W+1}w_{N_W-1,l}^{\textrm{in}}m_{t-(N_W+1)}+\cdots+w^{2N_W}w_{N_W,l}^{\textrm{in}}m_{t-2N_W}\\
                  &+w^{2N_W+1}w_{N_W-1,l}^{\textrm{in}}m_{t-(2N_W+1)}+\cdots,
    \end{aligned}
  \end{displaymath}
  where $w$ is an element of the recurrent matrix $\boldsymbol{W}$ which is assumed to be equal for all of the ESN models.  The output weight matrix of each ESN model $l$ can be given by
  \begin{displaymath}
    {\boldsymbol{W}^\textrm{out}_l}={\boldsymbol{R}_l}^{-1}\boldsymbol{p}_{k,l},
  \end{displaymath}
  where ${\boldsymbol{R}_l}={ {\mathbb E}}\left[\boldsymbol{\mu}_t^{\left(l\right)}\left({\boldsymbol{\mu}}_t^{\left(l\right)}\right)^{\rm T}\right]$ represents the covariance matrix with $\boldsymbol{\mu}_t^{\left(l\right)}=\left[ \mu_{1,t}^{\left(1\right)},\ldots, \mu_{N_W,t}^{\left(l\right)}\right]$  and $\boldsymbol{p}_{k,l}={\mathbb E}\left[\boldsymbol{\mu}^{\left(l\right)}_tm_{t-k}\right]$.
The element $\boldsymbol{{R}}_{l,12}$ in ${\boldsymbol{R}}_l$ can be calculated as follows
  \begin{equation}\small\label{eq:R}
    \begin{aligned}
        \boldsymbol{{R}}_{l,12}=&{\mathop {\mathbb E}}\left[\mu^{\left(l\right)}_{2,t}\mu^{\left(l\right)}_{1,t}\right]
        				 ={\mathop {\mathbb E}}[w_{1,l}^{\textrm{in}}w_{2,l}^{\textrm{in}}m^2_t+w^2w_{N_W,l}^{\textrm{in}}w_{1,l}^{\textrm{in}}m^2_{t-1}+\cdots 
                         +w^{2(N_W-1)}w_{2,l}^{\textrm{in}}w_{3,l}^{\textrm{in}}m^2_{t-(N_W-1)}\\
                         &+w^{2N_W}w_{1,l}^{\textrm{in}}w_{2,l}^{\textrm{in}}m^2_{t-N_W}+\cdots +w^{2(2N_W-1)}w_{2,l}^{\textrm{in}}w_{3,l}^{\textrm{in}}m^2_{t-(2N_W-1)}
                         +w^{4N_W}w_{1,l}^{\textrm{in}}w_{2,l}^{\textrm{in}}m^2_{t-2N_W}+\cdots ]\\
                         =&\sigma^2(w_{1,l}^{\textrm{in}}w_{2,l}^{\textrm{in}}+w^2w_{N_W,l}^{\textrm{in}}w_{1,l}^{\textrm{in}}+\cdots +w^{2(N_W-1)}w_{2,l}^{\textrm{in}}w_{3,l}^{\textrm{in}}
                         +w^{2N_W}w_{1,l}^{\textrm{in}}w_{2,l}^{\textrm{in}}+w^{2(N_W+1)}w_{N_W,l}^{\textrm{in}}w_{1,l}^{\textrm{in}}+\cdots \\
                         &+w^{2(2N_W-1)}w_{2,l}^{\textrm{in}}w_{3,l}^{\textrm{in}}+w^{4N_W}w_{1,l}^{\textrm{in}}w_{2,l}^{\textrm{in}}+\cdots )\\
                         =&\sigma^2\sum_{j=0}^{\infty}w^{2N_Wj}\left(w_{1,l}^{\textrm{in}}w_{2,l}^{\textrm{in}}+w^2w_{N_W,l}^{\textrm{in}}w_{1,l}^{\textrm{in}}+\cdots 
                         +w^{2(N_W-1)}w_{2,l}^{\textrm{in}}w_{3,l}^{\textrm{in}}\right)\\
                         =&\frac{\sigma ^2}{1-w^{2N_W}}\left({\mathop {\rm {rot}}}_1\left(\boldsymbol{w}^\textrm{in}_{N_W\ldots 1,l}\right)\right)^{\rm T}\Gamma ^2{\mathop {\rm {rot}}}_2(\boldsymbol{w}^\textrm{in}_{N_W \ldots 1,l}),
    \end{aligned}
  \end{equation}
  where $\sigma ^2$ is the variance of the input signal $m_t$. $\boldsymbol{w}^\textrm{in}_{N_W \ldots 1,l}=\left[ w_{N_W,l}^{\textrm{in}}, w_{N_W-1,l}^{\textrm{in}}, \ldots, w_{1,l}^{\textrm{in}} \right]$ and ${\mathop {\rm {rot}}}_k\left( \boldsymbol{w}^\textrm{in}_{N_W\ldots 1,l} \right)$ denote an operator that rotates vector $\boldsymbol{w}^\textrm{in}_{N_W\ldots 1,l}$ by $k$ place to the right. For example,  ${\mathop {\rm {rot}}}_1\left( \boldsymbol{w}^\textrm{in}_{N_W\ldots 1,l} \right)=\left[ w_{1,l}^{\textrm{in}}, w_{N_W,l}^{\textrm{in}}, \ldots, w_{2,l}^{\textrm{in}} \right]$. $\Gamma=\diag \left(1,w, \ldots, w^{N_W-1}\right)$. The element $  \boldsymbol{{R}}_{l,ij}$ can be given by:   
  \begin{equation}
  \setlength{\abovedisplayskip}{3 pt}
\setlength{\belowdisplayskip}{3 pt}
    {\boldsymbol{R}}_{l,ij}=\frac{\sigma ^2}{1-w^{2N_W}}({\mathop {\rm {rot}}}_i(\boldsymbol{w}^\textrm{in}_{N_W\ldots 1,l}))^{\rm T}\Gamma ^2{\mathop {\rm {rot}}}_j(\boldsymbol{w}^\textrm{in}_{N_W\ldots 1,l}).
  \end{equation}
  Thus, ${ \boldsymbol{R}_l}=\frac{\sigma ^2}{1-w^{2N_W}}\boldsymbol{V}_l^{\rm T}\Gamma ^2\boldsymbol{V}_l=\frac{\sigma ^2}{1-w^{2N_W}}{\boldsymbol{A}}_l\label{parallel_R}$ where ${\boldsymbol{A}}_l=\boldsymbol{V}_l^{\rm T}\Gamma ^2\boldsymbol{V}_l $. Similarly, based on (\ref{eq:R}),  element $p_{k1,l}$ is given by:
  \begin{equation}
    \setlength{\abovedisplayskip}{3 pt}
\setlength{\belowdisplayskip}{3 pt}
  	\begin{aligned}
  		p_{k1,l}=&{\mathop {\mathbb E}}\left[{\mu}_{1,t}^{\left(l\right)}m_{t-k}\right]
  	                     =&\sigma ^2w^{k}w_{\left(N_W-k+1\right){\mathop {\rm mod}N_W},l}^\textrm{in}.
    \end{aligned}
  \end{equation}
  We assume that $w_{0,l}^\textrm{in}=w_{N_W,l}^\textrm{in}$. Hence, $\boldsymbol{p}_{k,l}=\sigma ^2w^{k}{\mathop {\rm rot}}_k\left(\boldsymbol{w}_{1\ldots N_W,l}^\textrm{in}\right)\label{parallel_pk}{\rm .}$ Then, the output weight matrix of each ESN model $l$ is ${\boldsymbol{W}}_l^\textrm{out}=(1-w^{2N_W})w^{k}{\boldsymbol A}_{l}^{-1}{\mathop {\rm rot}}_{k}(\boldsymbol{w}_{1\ldots N_W,l}^\textrm{in})$. 

  The output of each ESN model $l$ at time $t$ can be given by $z_{l,t}=(\boldsymbol{\mu}^{\left(l\right)}_{t})^{\rm T} {\boldsymbol W}_l^\textrm{out}=(1-w^{2N_W})w^{k} (\boldsymbol{\mu}^{\left(l\right)}_{t})^{\rm T}  {\boldsymbol A}_{l}^{-1}{\mathop {\rm rot}}_{k}(\boldsymbol{w}_{1\ldots N_W,l}^\textrm{in})$ and $z_t=\sum_{l=1}^{L}z_{l,t}$. Then, the covariance of the output with the $k$-slot delayed input can be calculated by:
  \begin{equation}\label{eq:cov}
  	\begin{aligned}
  		\textrm{Cov}(z_t,m_{t-k})=&\sum_{l=1}^{L}(1-w^{2N_W})w^{k}\textrm{Cov}\left(\left(\boldsymbol{\mu}^{\left(l\right)}_t \right)^{\rm T},m_{t-k}\right) {\boldsymbol A}_{l}^{-1}{\mathop {\rm rot}}_{k}(\boldsymbol{w}_{1\ldots N_W,l}^\textrm{in})\\									
  										=&L(1-w^{2N_W})w^{2k}\sigma ^2\left({\mathop {\rm rot}}_{k}(\boldsymbol{w}_{1\ldots N_W,l}^\textrm{in})\right)^{\rm T} {\boldsymbol A}_{l}^{-1}{\mathop {\rm rot}}_{k}(\boldsymbol{w}_{1\ldots N_W,l}^\textrm{in})\\  						
  										\mathop  = \limits^{\left( a \right)}  &L(1-w^{2N_W})w^{2k}\sigma ^2\zeta _k,
  	\end{aligned}
  \end{equation}
  where (a) is obtained from the fact that $ \zeta _k= \left({\mathop {\rm rot}}_{k}(\boldsymbol{w}_{1\ldots N_W,l}^\textrm{in})\right)^{\rm T} {\boldsymbol A}_{l}^{-1}{\mathop {\rm rot}}_{k}(\boldsymbol{w}_{1\ldots N_W,l}^\textrm{in})$. The variance of the output can be given by:
    \begin{equation}\label{eq:var}
  	\begin{aligned}
  		\textrm{Var}(z_t)=&\mathbb{E}\left[ \sum_{l=1}^{L} z_{l,t} \sum_{p=1}^{L} z_{p,t} \right]-\left( \mathbb{E}\left[\sum_{l=1}^{L} z_{l,t} \right]\right)^2= \sum_{l=1}^{L}\sum_{p=1}^{L} \mathbb{E}\left[z_{l,t} z_{p,t} \right].\\
	  	\end{aligned}
  \end{equation}	
Since $\mathbb{E}\left[z_{p,t}z_{l,t} \right]=(1-w^{2N_W})w^{2k}\sigma ^2\zeta _k$, we have $\textrm{Var}(z_t)=L^2(1-w^{2N_W})w^{2k}\sigma ^2\zeta _k$.		


 The memory capacity of the parallel ESN model  can be given by
  \begin{equation}\label{eq:M}\small
  	\begin{aligned}
	  	M&= \sum_{k=1}^{\infty}\frac{\textrm{Cov}^2(m_{t-k},z_t)}{\textrm{Var}(m_t)Var(z_t)}=\frac{1}{L}\sum_{k=0}^{\infty} \frac{\textrm{Cov}(m_{t-k},z_t)}{\sigma^2}-\frac{\textrm{Cov}(m_{t},z_t)}{L\sigma^2}\\
	  	  &=\frac{1}{L}\sum_{k=0}^{\infty}\sum_{l=1}^{L}(1-w^{2N_W})w^{2k}\zeta _k-\frac{1}{L}\sum_{l=1}^{L}(1-w^{2N_W})\zeta _0\\
		  &=\frac{1}{L}\sum_{l=1}^{L}(1-w^{2N_W})\sum_{k=0}^{\infty} w^{2k}\zeta _k - \frac{1}{L}\sum_{l=1}^{L}(1-w^{2N_W})\zeta _0   \\
		  &= \frac{1}{L}\sum_{l=1}^{L}(1-w^{2N_W}) \left[ \sum_{k=0}^{N_W-1} w^{2k}\zeta _k+\sum_{k=N_W}^{2N_W-1} w^{2k}\zeta _k + \dots \right] -\frac{1}{L} \sum_{l=1}^{L}(1-w^{2N_W})\zeta _0 \\
		  &=\frac{1}{L} \sum_{l=1}^{L}(1-w^{2N_W}) \left(\sum_{k=0}^{N_W-1} w^{2k}\zeta _k \right) \left( \sum_{k=0}^{\infty} w^{2N_Wk} \right)- \frac{1}{L}\sum_{l=1}^{L}(1-w^{2N_W})\zeta _0\\
		  &= \frac{1}{L}\sum_{l=1}^{L}\left(\sum_{k=0}^{N_W-1} w^{2k}\zeta _k-(1-w^{2N_W})\zeta _0 \right) \mathop  = \limits^{\left( a \right)}  \frac{1}{L}\sum_{l=1}^{L}\left(\sum_{k=1}^{N_W-1} w^{2k}\zeta _k+w^{2N_W}\zeta _0 \right) \\
	  	  &=\frac{1}{L}\sum_{l=1}^{L}\left(\sum_{k=1}^{N_W-1} w^{2k}\zeta _k+w^{2N_W}\zeta _{N_W} \right)=\frac{1}{L}\sum_{l=1}^{L}\left(\sum_{k=1}^{N_W} w^{2k}\zeta _k \right)\mathop  = \limits^{\left( b \right)} N_W-1+w^{2N_W},
  	\end{aligned}
  \end{equation}
  where (a) is obtained from the fact that $\zeta _0=\zeta _{N_W}$ and (b) stems from the fact that $w_l^{2k}\zeta _k=1$ as $k=1,\ldots, N_W-1$ and $w^{2N_W}\zeta _{N_W}=w^{2N_W}$. This completes the proof.
  
\subsection{Proof of Theorem \ref{th:2}}\label{Ap:b} 

Let $\boldsymbol{m}_{\ldots t}=\left[m_1,\ldots,m_{t-1},m_t \right]$ be the input steam vector and $z_{t}^{\left(l\right)}$ be the output of ESN model $l$. Next, we derive the memory capacity of a series ESN model using an enumeration method. First, according to (\ref{eq:mu})-(\ref{eq:var}), we have  $\textrm{Cov}(z_t^{\left(1\right)},m_{t-k})=(1-w^{2N_W})w^{2k}\sigma ^2\zeta _k$. Given the output of the first ESN model, $z_t^{\left(1\right)}$, which is the input of the second ESN model, then $\textrm{Cov}(z_t^{\left(2\right)},m_{t-k})=  \left(1-w^{2N_W}\right)^2w^{2k} \zeta _k\sigma ^2 $. Similarly, we can obtain that $\textrm{Cov}(z_t^{\left(3\right)},m_{t-k})=  \left(1-w^{2N_W}\right)^3w^{2k} \zeta _k  \sigma ^2 $. Therefore, we can conclude that $\textrm{Cov}(z_t^{\left(L\right)},m_{t-k})= \left(1-w^{2N_W}\right)^L w^{2k} \zeta _k$. Based on (\ref{eq:M}), the memory capacity of a series ESN can be given by $\sum_{k=0}^{\infty} \left( 1-w^{2N_W}\right)^Lw^{2k} \zeta _k  - \left( 1-w^{2N_W}\right)^L\zeta _0$. Then, the memory capacity of a series ESN is 
 \begin{equation}\label{eq:M}\small
  	\begin{aligned}
	  	M&=\sum_{k=0}^{\infty}  \left(1-w^{2N_W}\right)^Lw^{2k} \zeta _k -  \left(1-w^{2N_W}\right)^L \zeta _0 = \left( 1-w^{2N_W}\right)^L \sum_{k=0}^{\infty}  w^{2k}\zeta _k - \left(1-w^{2N_W}\right)^L \zeta _0 \\
	 &= \left( 1-w^{2N_W}\right)^L \left[  \sum_{k=0}^{N_W-1}   w^{2k}\zeta _k  +\sum_{k=N_W}^{2N_W-1} w^{2k}\zeta _k+ \dots\right] - \left(1-w^{2N_W}\right)^L \zeta _0 \\	  
		  &= \left( 1-w^{2N_W}\right)^L \left(\sum_{k=0}^{N_W-1} w^{2k}\zeta _k  \right) \left( \sum_{k=0}^{\infty} w^{2N_Wk} \right)-  \left(1-w^{2N_W}\right)^L \zeta _0 \\
		  &=  \left( 1-w^{2N_W}\right)^L \left( \frac{\sum_{k=0}^{N_W-1} \left(w^{2k}\zeta _k\right)^L}{{1-w^{2N_W}}}-\frac{\zeta_0-\zeta_0w^{2N_W}}{1-w^{2N_W}} \right) = {\left(1-w_l^{2N_W}\right)}^{L-1}\left(N_W-1+{w^{2N_W}}\right).
  	\end{aligned}
  \end{equation}
 This completes the proof.
 
 \subsection{Proof of Theorem \ref{th:3}}\label{Ap:c} 
 
 The memory capacity of a single ESN with multiple inputs is derived using an enumeration method. Consider $K=2$, then the input stream will be $\boldsymbol{m}_{\ldots t}=\left[\boldsymbol{m}_{\ldots 1t},\boldsymbol{m}_{\ldots 2t} \right]$ where $\boldsymbol{m}_{\ldots kt}=m_{k1}\ldots m_{kt-2}m_{kt-1}m_{kt}$. Based on the proof of Theorems \ref{th:1} and \ref{th:2},  ${\boldsymbol R}$ can be given by: 
  \begin{equation}
  	\begin{aligned}
  	{\boldsymbol R}&=\frac{\sigma _1^2+2\rho_{12} \sigma _1\sigma _2+\sigma _2^2}{1-w^{2N_W}}\boldsymbol{V} ^{\rm T}\Gamma ^2\boldsymbol{V} =\frac{\sigma _1^2+2\rho_{12} \sigma _1\sigma _2+\sigma _2^2}{1-w^{2N_W}}{\boldsymbol A}\label{multi_R},
  	\end{aligned}
  \end{equation}
  and $\boldsymbol p_k=w^k\left(\sigma _1^2+\sigma _2^2\right)\mathop{\rm rot}\nolimits _k\left(\boldsymbol V_{1\ldots N}\right)\label{multi_pk}$. Then the output weight matrix can be given by:
  \begin{equation}
	{\boldsymbol W}^\textrm{out}=\left(1-w^{2N_W}\right)w^{k}\frac{\sigma _1^2+\sigma _2^2}{\sigma _1^2+2\rho_{12} \sigma _1\sigma _2+\sigma _2^2}{\boldsymbol A}^{-1}{\mathop {\rm rot}}_{k}\left(\boldsymbol V_{1\ldots N}\right){\rm .}
  \end{equation}
  The output at time $t$ is
  \begin{equation}
  	\begin{aligned}
  		z_t=\left(\boldsymbol{m}_t\right)^{\rm T}{\boldsymbol W}^\textrm{out}=\left(1-w^{2N_W}\right)w^{k}\frac{\sigma _1^2+\sigma _2^2}{\sigma _1^2+2\rho_{12} \sigma _1\sigma _2+\sigma _2^2}\left(\boldsymbol{m}_t\right)^{\rm T}{\boldsymbol A}^{-1}{\mathop {\rm rot}}_{k}\left(\boldsymbol V_{1\ldots N}\right)\label{multi_yt}.
  	\end{aligned}
  \end{equation}
The covariance of the output at time $t$ and $t-k$ can be given by:
  \begin{equation}\nonumber
  	\begin{aligned}
  		\textrm{Cov}&\left(z_t,\boldsymbol{m}_{t-k}\right)
  			=\left(1-w^{2N_W}\right)w^{k}\frac{\sigma _1^2+\sigma _2^2}{\sigma _1^2+2\rho_{12} \sigma _1\sigma _2+\sigma _2^2}\textrm{Cov}\left(\left(\boldsymbol{\mu}_t\right)^{\rm T}, \boldsymbol{m}_{t-k}\right)
  			\times {\boldsymbol A}^{-1}{\mathop {\rm rot}}_{k}\left(\boldsymbol V_{1\ldots N}\right)\\
  			=&\left(1-w^{2N_W}\right)w^{2k}\frac{\left(\sigma _1^2+\sigma _2^2\right)^2}{\sigma _1^2+2\rho_{12} \sigma _1\sigma _2+\sigma _2^2}\zeta _{k}{\rm .}
  	\end{aligned}
  \end{equation}	
Based on (\ref{eq:var}),  $\textrm{Var}\left(z_t\right)=\textrm{Cov}\left(z_t, \boldsymbol{m}_{t-k}\right).$ Then, the memory capacity can be given by

  \begin{equation}\small
  	\begin{aligned}
	  	M=&\sum_{k=0}^{\infty} \frac{\textrm{Cov}^2\left(z_t, \boldsymbol{m}_{t-k}\right)}{\textrm{Var}\left(\boldsymbol{m}_{t-k}\right)\textrm{Var}\left(z_t\right)} -\frac{\textrm{Cov}^2\left(z_t, \boldsymbol{m}_{t}\right)}{\textrm{Var}\left(\boldsymbol{m}_{t}\right)\textrm{Var}\left(z_t\right)} 
	  	  =\sum_{k=0}^{\infty} \frac{  \textrm{Cov}\left(z_t, \boldsymbol{m}_{t-k}\right)  }{\sigma _1^2+2\rho_{12}\sigma _1\sigma _2+\sigma _2^2}-\frac{  \textrm{Cov}\left(z_t, \boldsymbol{m}_{t}\right)  }{\sigma _1^2+2\rho_{12}\sigma _1\sigma _2+\sigma _2^2}\\
	  	  =&\left( \frac{\sigma _1^2+\sigma _2^2}{\sigma _1^2+2\rho_{12} \sigma _1\sigma _2+\sigma _2^2}\right) ^2\times\left(1-w^{2N_W}\right)\sum_{k=0}^{N_W-1}w^{2k}\zeta _k\sum_{j=0}^{\infty}r^{2N_Wj}-\left(1-w^{2N_W}\right)\\
	  	  =&\left( \frac{\sigma _1^2+\sigma _2^2}{\sigma _1^2+2\rho_{12} \sigma _1\sigma _2+\sigma _2^2}\right) ^2\sum_{k=0}^{N_W-1}w^{2k}\zeta _k-(1-w^{2N_W})=\left( \frac{\sigma _1^2+\sigma _2^2}{\sigma _1^2+2\rho_{12}\sigma _1\sigma _2+\sigma _2^2}\right) ^2\left(N_W-1+w^{2N_W}\right).
  	\end{aligned}
  \end{equation}
  Similarly, we can formulate the memory capacity of the single ESN with input vector $\boldsymbol{m}_t=\left[\boldsymbol{m}_{1t},\boldsymbol{m}_{2t},\boldsymbol{m}_{3t}\right]$, which is given by $\left( \frac{\sigma _1^2+\sigma _2^2+\sigma _3^2}{\sigma _1^2+2\rho_{12}\sigma _1\sigma _2+2\rho_{13}\sigma _1\sigma _3+2\rho_{23}\sigma _2\sigma _3+\sigma _2^2+\sigma _3^2}\right) ^2\left(N_W-1+w^{2N_W}\right).$ In consequence, the memory capacity of a single ESN with input vector $\boldsymbol{m}_t=\left[\boldsymbol{m}_{1t}, \ldots,\boldsymbol{m}_{Kt}\right]$ can be given by $ \left(\frac{\sum_{l=1}^{K}\sigma _l^2}{\sum_{k=1}^{K}\sum_{n=1}^{K}\rho _{kn}\sigma _k\sigma _n}\right)^2\left(N_W-1+w^{2N_W}\right)$. This completes the proof.

\bibliographystyle{IEEEbib}
\def\baselinestretch{1.38}
\bibliography{references}
\end{document}